\documentclass[sigconf,screen]{acmart}

\usepackage{rotating}
\usepackage{amsmath,amsfonts}
\usepackage{array}
\usepackage{diagbox}
\usepackage{booktabs} 
\usepackage{colortbl}   
\usepackage{multirow}
\usepackage{makecell}
\usepackage{tablefootnote}
\usepackage{arydshln}
\usepackage{tikz}
\usepackage{longtable}
\usepackage{makecell}
\usepackage{pifont}
\usepackage{algorithm}
\usepackage{algpseudocode}

\AtBeginDocument{%
  }

\setcopyright{rightsretained}
\copyrightyear{2026}
\acmYear{2026}
\setcopyright{cc}
\setcctype{by}
\acmConference[Conference acronym 'XX]{Make sure to enter the correct
  conference title from your rights confirmation email}{June 03--05,
  2018}{Woodstock, NY}
\acmBooktitle{}

\begin{document}

\title{Relevance-Guided Residual Calibration for Spectral Compression of Dense Retrieval Embeddings}

\author{Yongkang Li}
\email{y.li7@uva.nl}
\orcid{0000-0001-6837-6184}
\affiliation{%
  \institution{University of Amsterdam}
  \city{Amsterdam}
  \country{The Netherlands}
}

\author{Panagiotis Eustratiadis}
\orcid{0000-0002-9407-1293}
\email{p.efstratiadis@uva.nl}
\affiliation{%
  \institution{University of Amsterdam}
  \city{Amsterdam}
  \country{The Netherlands}
}

\author{Evangelos Kanoulas}
\orcid{0000-0002-8312-0694}
\email{e.kanoulas@uva.nl}
\affiliation{%
  \institution{University of Amsterdam}
  \city{Amsterdam}
  \country{The Netherlands}
}



\begin{abstract}
Spectral compression reduces the storage and scoring costs of dense retrieval by projecting embeddings onto a smaller principal subspace. However, its corpus-derived weighting rules tie all retained directions to a shared spectral trend, without directly accounting for their individual contributions to relevance.
To address this, we propose \emph{Relevance-Guided Residual Calibration} (RRC), a closed-form method for refining spectral weights using relevance judgments.
RRC estimates margin-based weights from the means and variances of positive--negative score contributions along each retained direction. 
It retains directions with positive estimated weights, decomposes their log weights into a fitted power-law trend and a residual, and adds the residual to the compressor's log weights on this support.
This allows relevance judgments to adjust individual directions while preserving the compressor's overall spectral trend. RRC applies to PCA, Whitening, and SpecTemp, and incorporates the corrections into the existing transformation shared by documents and queries, without iterative model training or an additional inference stage.
Experiments with three encoders on nine datasets covering conventional and reasoning-intensive retrieval show that 
SpecTemp\,+\,RRC achieves the highest mean effectiveness among tested compressors on the BEIR group and BRIGHT at every target dimension.
On the five evaluated BRIGHT datasets, it improves over SpecTemp by an average of 2.34 nDCG@10 points at 768 dimensions. The gains also persist under scalar and product quantization, improving the trade-off between per-document storage and retrieval effectiveness. Our code is available at: \url{https://anonymous.4open.science/r/rrc-spectral-compression-B5CE/}
\end{abstract}




\maketitle

\section{Introduction}

Dense retrieval supports semantic search over Web content by representing queries and documents as vectors and ranking documents by their similarity~\cite{KarpukhinOMLWEC20_DPR,LinALOLMY023_dragon}.
Recent encoders built on large language models produce embeddings~\cite{long2025diver,li2023gte_7b} with thousands of dimensions; for example, Qwen3-Embedding-8B and Harrier-27B produce embeddings of 4096 and 5376 dimensions, respectively~\cite{zhang2025qwen3embedding,harrier2026}.
For Web search systems, the storage and similarity-computation costs of these representations grow with both the embedding dimensionality and the collection size. Embedding compression addresses this deployment challenge by reducing representation size while preserving the information needed to rank relevant documents. We focus on post-hoc compression of existing embeddings, a practical setting when the encoder is fixed or accessible only through an API.\looseness=-1

Within this setting, learned compression methods, including conditional autoencoders~\cite{liu-etal-2022-dimension} and Matryoshka-Adaptor~\cite{yoon-etal-2024-matryoshka}, optimize transformations over existing embeddings. 
These methods offer flexible mappings but introduce additional training costs and optimization choices, such as learning rates and stopping criteria. For supervised variants, scarce relevance judgments can further constrain validation and model selection.
Spectral methods construct a fixed compression map directly from corpus statistics~\cite{raunak-etal-2019-effective}: they project embeddings onto the $k$ leading principal directions and rescale the projected coordinates.
Common choices assign weights proportional to $\lambda_i^{-\gamma}$ to the retained spectral directions, where $\lambda_i$ is the $i$-th eigenvalue and the shared exponent $\gamma\in[0,1]$ controls whitening degree~\cite{kexuefm-9079}. Varying $\gamma$ interpolates between PCA ($\gamma=0$), which leaves the projected coordinates unscaled~\cite{zhang-etal-2024-evaluating-unsupervised}, and Whitening ($\gamma=1$), which equalizes their variances~\cite{su2021whiteningsentence, huang-etal-2021-whiteningbert-easy}.
SpecTemp~\cite{LiEK26_SpecTemp} selects this exponent from the corpus spectrum and the target dimension.
These methods yield a transformation that requires no iterative model training and apply identically to document and query embeddings. Their weighting rules, however, are derived from corpus statistics without using relevance judgments for the target retrieval task.\looseness=-1 

Corpus variance describes how much a projected coordinate varies across documents, but does not directly measure whether its associated direction helps distinguish relevant documents from competing candidates. Within the retained principal subspace, a shared exponent shifts weight between high- and low-variance directions, while tying all weights to the same power-law trend. This rule cannot independently strengthen or suppress a direction according to its contribution to relevance. Relevance judgments provide evidence about these direction-specific contributions. Learned adaptors exploit such supervision to improve retrieval with fixed embeddings~\cite{yoon-etal-2024-search}, including compressed representations~\cite{yoon-etal-2024-matryoshka}.
We investigate how this relevance information can directly calibrate spectral weights through a closed-form update, retaining a single fixed compression map for document indexing and query compression.\looseness=-1

We propose \emph{Relevance-Guided Residual Calibration} (RRC), which refines the weights assigned by an existing spectral compressor using relevance judgments. We refer to the compressor's weights before calibration as the \emph{spectral prior}. RRC first estimates margin-based weights from training triples, accounting for each direction's mean contribution to the positive--negative score margin and its variance. These estimates contain both an overall spectral trend and direction-specific departures from that trend. RRC separates the two by fitting a power law to the positive estimated weights in log space. 
It then adds the residual to the prior's log weights on the positive support and sets the remaining weights to zero.
For a fixed spectral basis and training triples, the same residual can calibrate PCA, whitening, or SpecTemp. Calibration requires only moment estimation and a least-squares fit, and the resulting weights are incorporated into the existing compression map without adding a 
separate inference stage.

We evaluate RRC with three LLM-based encoders on nine datasets covering conventional and reasoning-intensive retrieval, at four target dimensions. 
SpecTemp\,+\,RRC achieves the highest mean effectiveness among the evaluated compression methods in both BEIR and BRIGHT dataset groups at every target dimension.
Its calibration gains are larger on the five evaluated BRIGHT datasets~\cite{su2024bright}, where it improves over SpecTemp by an average of 2.34 nDCG@10 points at a target dimension of 768.
Ablations show that adding RRC improves mean effectiveness at every tested power-law exponent. With SpecTemp as the prior, support selection and residual weighting both raise mean effectiveness, and retaining the prior's exponent yields a higher mean than directly using the estimated margin-based weights. Finally, experiments with scalar and product quantization further show that the calibration gains persist, improving the trade-off between per-document storage and retrieval effectiveness.\looseness=-1

Our contributions are threefold:

$\bullet$ We introduce RRC, a closed-form method that calibrates spectral compression using relevance judgments, without iterative model training, while retaining a single fixed compression map for document indexing and query compression.

$\bullet$ We develop a trend--residual decomposition that makes label-derived corrections reusable across different spectral weighting rules. The decomposition preserves the chosen prior's spectral trend and leaves the original weights unchanged when the estimated weights follow a power law over the leading $k$ principal directions.

$\bullet$ We analyse support selection, residual weighting, and spectral trend choice, and relate power-law fit quality to calibration gains. Experiments show mean effectiveness gains across tested priors and improved storage--effectiveness trade-offs under quantization.
\section{Related Work}
\label{sec:related-work}

Embedding compression seeks compact representations that preserve retrieval effectiveness. We review approaches to learning and encoding these representations, spectral transformations of existing embeddings, and relevance-aware coordinate weighting.

\subsection{Embedding Compression for Dense Retrieval}

Compact representations can be learned alongside an encoder or constructed from its existing outputs. Matryoshka Representation Learning~(MRL) trains nested representations with informative prefixes, enabling compression by prefix truncation~\cite{KusupatiBRWSRHC22_MRL}. With a fixed encoder, ConAE learns a conditional autoencoder for retrieval-oriented dimensionality reduction~\cite{liu-etal-2022-dimension}. Search-Adaptor uses relevance supervision to adapt existing embeddings~\cite{yoon-etal-2024-search}, while Matryoshka-Adaptor learns compressed representations at multiple dimensions through supervised or unsupervised objectives~\cite{yoon-etal-2024-matryoshka}. SMEC combines sequential Matryoshka training with adaptive dimension selection and cross-batch learning~\cite{zhang-etal-2025-smec}. CSR learns sparse codes over frozen embeddings as an alternative to Matryoshka-style truncation~\cite{pmlr-v267-wen25e}.
These approaches learn how information is distributed across output coordinates; RRC instead calibrates weights in a fixed principal subspace without iterative model training.\looseness=-1

Simpler alternatives reduce dimensionality without learning a compression module. Random projection uses a data-independent linear map~\cite{johnson1984extensions}, and random coordinate selection has been reported to preserve much of the retrieval effectiveness of the full embeddings~\cite{takeshita-etal-2025-randomly}.
Quantization is complementary to dimensionality reduction: it reduces storage by encoding vectors more compactly.
Product quantization~(PQ) encodes subvectors using separate codebooks~\cite{jegou2010product}, optimized product quantization~(OPQ) optimizes the space decomposition and codebooks~\cite{Ge2013OPQ}, and RaBitQ provides error-bounded distance estimation through randomized quantization~\cite{Gao2024RaBitQ}. PCA and PQ have also been combined in a retrieval compression pipeline~\cite{ma-etal-2021-simple}, and Section~\ref{sec:exp-quantization} examines how relevance calibration affects such combinations.

 \subsection{Spectral Transformations}

Unlike random maps, spectral compressors use the covariance structure of a collection to select principal directions and rescale the projected coordinates.
A common family projects embeddings onto the leading principal directions and assigns them weights proportional to $\lambda_i^{-\gamma}$, with PCA ($\gamma=0$) and whitening ($\gamma=1$) as endpoints and intermediate values giving partial whitening~\cite{kexuefm-9079}. PCA preserves the variances of the projected coordinates~\cite{raunak-etal-2019-effective,zhang-etal-2024-evaluating-unsupervised}, whereas whitening equalizes them and can amplify noise in low-variance directions~\cite{su2021whiteningsentence,huang-etal-2021-whiteningbert-easy};
Soft-ZCA uses eigenvalue regularization to control whitening strength~\cite{Diera2025SoftZCA}.
SpecTemp selects the exponent from the corpus spectrum and the target dimension~\cite{LiEK26_SpecTemp}.
PCA can also improve retrieval effectiveness in some domain-adaptation settings~\cite{zuo2026efficiencyembeddingcompress}.
In each of these weighting rules, the weight of a retained direction is a fixed function of its eigenvalue, derived from corpus statistics without relevance judgments for the target retrieval task. RRC uses these judgments to calibrate direction-specific weights while retaining the prior's spectral trend.

\subsection{Relevance-Aware Coordinate Weighting}

For retrieval, a useful coordinate is one that helps distinguish relevant documents from competing candidates, and several methods estimate this usefulness from feedback.
Dimension importance estimation~(DIME) uses pseudo-relevance feedback or generated pseudo-documents to select a query-dependent subspace of the original embeddings~\cite{Faggioli2024DIME,Faggioli2026GettingOffDIME}. Follow-up work examines its reproducibility and feedback aggregation~\cite{Campagnano2025UnveilingDIME}, contrasts positive with pseudo-irrelevant feedback~\cite{DErasmo2025Eclipse}, and determines the retained dimensions from a risk criterion~\cite{d-erasmo-etal-2026-statistical}. Learning to Select trains a predictor of dimension importance from labeled positive and negative documents and masks query coordinates at inference~\cite{wu-etal-2026-learning}. 
These methods select coordinates separately for each query while retaining full-dimensional document representations. RRC aggregates relevance evidence across training queries into shared weights in a fixed spectral basis, producing a single map that compresses both document and query embeddings.

The margin-based weights estimated by RRC (Section~\ref{sec:margin-reference}) have a mean--variance structure related to diagonal linear discriminant rules, which weight class-mean differences by inverse feature variances~\cite{Bickel2004Fisher,Fan2008FAIR}. 
RRC applies this weighting to coordinate-level retrieval margins and decomposes the estimated log weights into a fitted power-law trend and a direction-specific residual. Transferring the residual to a spectral prior separates the two sources of information: the chosen compressor provides the overall spectral trend, and relevance judgments supply the direction-specific adjustments.\looseness=-1 
\section{Methodology}
\label{sec:methodology}

Given fixed embeddings $\mathbf{x}_t\in\mathbb{R}^{D}$, where $t$ denotes a query or document, we seek a shared mapping $T_k:\mathbb{R}^{D}\rightarrow\mathbb{R}^{k}$, with $k<D$, that reduces representation size while preserving retrieval effectiveness. We consider post-hoc compression, keeping the embedding encoder frozen.
We first formulate spectral compression through direction-wise weights, then introduce RRC to calibrate these weights using relevance judgments from a small set of training queries.

\subsection{Spectral Compression and Prior Weights}
\label{sec:spectral-compression}
 
From $n$ sampled corpus embeddings, we compute the mean $\boldsymbol{\mu}$ and eigendecompose the centred covariance:
\begin{equation}
    \mathbf{C}
    =\frac{1}{n-1}\sum_{j=1}^{n}
      (\mathbf{x}_{d_j}-\boldsymbol{\mu})
      (\mathbf{x}_{d_j}-\boldsymbol{\mu})^{\top}
    =\mathbf{U}\boldsymbol{\Lambda}\mathbf{U}^{\top},
    \label{eq:covariance}
\end{equation}
where $\boldsymbol{\Lambda}=\operatorname{diag}(\lambda_1,\ldots,\lambda_D)$ and $\lambda_1\geq\cdots\geq\lambda_D\geq0$. 
For a target dimensionality $k \leq \operatorname{rank}(\mathbf C)$, the leading eigenvectors $\mathbf U_k \in \mathbb R^{D\times k}$ define a shared basis, yielding the projected coordinates
\begin{equation}
    \mathbf{c}(t)=\mathbf{U}_k^{\top}
      (\mathbf{x}_t-\boldsymbol{\mu})\in\mathbb{R}^{k}.
    \label{eq:coordinates}
\end{equation}
A nonnegative, nonzero score-weight vector $\mathbf{w}=(w_1,\ldots,w_k)^{\top}$ determines the compressed embedding and its inner-product score:
\begin{align}
    \mathbf{y}_{\mathbf{w}}(t)
    &=\operatorname{diag}(\sqrt{w_1},\ldots,\sqrt{w_k})\,\mathbf{c}(t),
    \label{eq:representation}\\
    S_{\mathbf{w}}(q,d)
    &=\mathbf{y}_{\mathbf{w}}(q)^{\top}\mathbf{y}_{\mathbf{w}}(d)
      =\sum_{i=1}^{k}w_i\,c_i(q)c_i(d).
    \label{eq:score}
\end{align}
Thus, $\sqrt{w_i}$ scales a coordinate and $w_i$ weights its score contribution. Cosine retrieval additionally normalises the compressed vectors. A common positive scaling of $\mathbf{w}$ leaves either ranking rule unchanged.
 
We call the positive weights of an existing spectral compressor its \emph{spectral prior}, denoted $\mathbf{p}$. Common priors belong to the power-law family\looseness=-1
\begin{equation}
    p_i(k)=\lambda_i^{-\gamma(k)},\qquad i=1,\ldots,k.
    \label{eq:priors}
\end{equation}
PCA uses $\gamma(k)=0$, whitening uses $\gamma(k)=1$, and fixed-exponent tempering uses $\gamma(k)=\gamma_0\in(0,1)$. 
SpecTemp~\cite{LiEK26_SpecTemp} instead selects $\gamma(k)$ from the corpus spectrum, applying full whitening up to the knee of a spectral signal-to-noise curve and progressively weaker whitening at larger target dimensions (see Appendix~\ref{app:spectemp}).
Without calibration, $\mathbf{w}=\mathbf{p}$.

For the power-law spectral priors in Eq.~\eqref{eq:priors}, a shared exponent ties the weights of the leading $k$ principal directions to the same spectral trend.
We next introduce RRC, which uses relevance judgments to complement this trend with direction-specific corrections.

\subsection{Relevance-Guided Residual Calibration}
\label{sec:residual-calibration}

We use relevance judgments from a small set of training queries to calibrate a spectral prior for unseen-query compression, keeping the spectral basis fixed. Since the common scale of the weights does not affect rankings, we adjust their relative magnitudes in log space.
This also makes a power-law trend linear in the log eigenvalues, allowing us to separate it from direction-specific residuals. We add the label-derived residual to the prior's log weights:
\begin{equation}
    \boxed{\;
    \log w_i(k)
    =\underbrace{\log p_i(k)}_{\text{spectral prior}}
     +\underbrace{u_i(k)}_{\text{label residual}},
    \quad i\in\mathcal{I}_{+}(k).
    \;}
    \label{eq:calibration}
\end{equation}
Here, $\mathcal{I}_{+}$ is the positive support of the estimated margin-based weights, and $u_i$ is their log-space residual after removing the fitted power-law trend. 
Equivalently, $w_i=p_i\,e^{u_i}$ for $i\in\mathcal{I}_{+}$, and $w_i=0$ outside this support.

For a fixed spectral basis and training triples, the same residual can calibrate different positive spectral priors, yielding PCA\,+\,RRC, SpecTemp\,+\,RRC, and analogous combinations. Below, we fix the target dimension $k$ and suppress its notation.

\subsubsection{A Margin-Based Weighting Criterion}
\label{sec:margin-reference}

For a triple $(q,d^+,d^-)$ consisting of a query, a judged relevant document, and a sampled negative, define the unweighted margin contribution of direction $i$:
\begin{equation}
    z_i=c_i(q)\bigl(c_i(d^{+})-c_i(d^{-})\bigr),
    \qquad
    m_i=\mathbb{E}[z_i],\quad v_i=\operatorname{Var}(z_i).
    \label{eq:coordinate-margin}
\end{equation}
The expectations use a query-balanced triple distribution~(defined in Section~\ref{sec:margin-estimation}).
The score margin of the triple under Eq.~\eqref{eq:score} is $\sum_i w_i z_i$, with mean $\sum_i w_i m_i$ and variance $\sum_{i,j}w_iw_j\operatorname{Cov}(z_i,z_j)$. Under a normal approximation, the probability that the relevant document outscores the negative increases monotonically with the ratio of this mean to its standard deviation.
To limit the number of second-order quantities estimated from the judged queries, we approximate the covariance of these margin contributions by its diagonal. This reduces the number of covariance parameters from $k(k+1)/2$ to $k$ and gives the standardised-margin surrogate:
\begin{equation}
    A(\mathbf{w})
    =\frac{\sum_{i=1}^{k}w_im_i}
           {\sqrt{\sum_{i=1}^{k}w_i^{2}v_i}}.
    \label{eq:criterion}
\end{equation}
We use this inner-product surrogate to estimate the weights that determine RRC's positive support and residual corrections.
The maximiser of $A(\mathbf{w})$ has a closed form, obtained through the Cauchy--Schwarz inequality (see Appendix~\ref{app:margin-optimum-proof}).

\begin{proposition}[Margin-Based Spectral Weighting]
\label{prop:margin-optimum}
Suppose $v_i>0$ for all $i=1,\ldots,k$ and $m_i>0$ for at least one direction. The maximisers of $A(\mathbf{w})$ over nonnegative, nonzero weights are given, up to a common positive factor, by
\begin{equation}
    e_i\propto\frac{(m_i)_{+}}{v_i},
    \qquad i=1,\ldots,k,
    \label{eq:margin-weights}
\end{equation}
and attain $A(\mathbf{e})=
\sqrt{\sum_{i=1}^{k}(m_i)_{+}^{2}/v_i}$.
\end{proposition}

These \emph{margin-based weights} favour directions with positive mean contributions and lower variability, while assigning zero weight to directions with nonpositive mean margins. We estimate them from training judgments, then extract their departures from a power-law trend to calibrate the chosen prior.

\subsubsection{Query-Balanced Estimation}
\label{sec:margin-estimation}

To estimate the margin means and variances in Proposition~\ref{prop:margin-optimum}, let $\mathcal{Q}_{\mathrm{tr}}$ contain $Q$ labeled training queries~(the calibration query budget).
For each query $q$, we construct $n_q$ positive--negative pairs from its relevance judgments and sampled negatives, and denote  the margin contribution of direction $i$ in the $j$-th pair by $z_{qj,i}$, with $j=1,\ldots,n_q$.
We first average within each query and then across queries, giving each query equal weight regardless of its number of sampled pairs. 
For a function $f$ of $z_i$, this query-balanced average is
\begin{equation}
    \widehat{\mathbb{E}}[f(z_i)]
    =\frac{1}{Q}\sum_{q\in\mathcal{Q}_{\mathrm{tr}}}\frac{1}{n_q}
      \sum_{j=1}^{n_q}f(z_{qj,i}),
    \label{eq:qb-average}
\end{equation}
and the moment estimates are
\begin{equation}
    \widehat{m}_i=\widehat{\mathbb{E}}[z_i],
    \qquad
    \widehat{v}_i=\widehat{\mathbb{E}}\bigl[(z_i-\widehat{m}_i)^2\bigr].
    \label{eq:est-moments}
\end{equation}
Eq.~\eqref{eq:est-moments} gives the plug-in estimates of the population
mean $m_i$ and variance $v_i$ in Eq.~\eqref{eq:coordinate-margin} under the
query-balanced distribution.
Substituting them into Eq.~\eqref{eq:margin-weights} gives the estimated margin-based weights\looseness=-1
\begin{equation}
\widehat{e}_i=\frac{(\widehat{m}_i)_{+}}{\widehat{v}_i}.
    \label{eq:empirical-weights-hat}
\end{equation}
The negative-sampling setup is detailed in Appendix~\ref{app:implementation-details}.\looseness=-1

\subsubsection{Separating the Spectral Trend from the Residual}
\label{sec:residual-decomposition}
We use the estimated margin-based weights $\widehat e_i$ in two ways: \emph{support selection} retains directions with positive estimates, and \emph{residual weighting} uses their log-space residuals to reweight the prior within this support.
Using $\widehat{e}_i$ directly would also replace the chosen compressor's overall spectral trend. We instead retain that trend and extract the direction-specific departures from a power law.

\emph{Support selection} defines the positive support $\mathcal{I}_{+}=\{i\leq k:\widehat{e}_i>0\}$.
On this support, we fit a power law to the estimated weights in log space, where the spectral trend becomes linear in $\log\lambda_i$.
Assuming the log eigenvalues vary across this support, least squares gives a unique intercept $a$ and exponent $\gamma_{\mathrm e}$:
\begin{equation}
    (a,\gamma_{\mathrm e}) =
    \operatorname*{arg\,min}_{a,\gamma\in\mathbb{R}} \sum_{i\in\mathcal{I}_{+}} \bigl(\log\widehat{e}_i-a+\gamma\log\lambda_i\bigr)^2.
    \label{eq:powerlaw-fit}
\end{equation}
Thus, $\exp(a)\lambda_i^{-\gamma_{\mathrm e}}$ is the least-squares power-law fit to the estimated weights on $\mathcal{I}_{+}$ in log space. 
We subtract this fitted trend to obtain direction-specific residuals, which we then combine with the chosen prior's log weights. 
The residual is
\begin{equation}
    u_i
    =\log\widehat{e}_i
     -\bigl(a-\gamma_{\mathrm e}\log\lambda_i\bigr),
    \qquad i\in\mathcal{I}_{+}.
    \label{eq:residual}
\end{equation}
The least-squares normal equations give
\begin{equation}
    \sum_{i\in\mathcal{I}_{+}}u_i=0,
    \qquad
    \sum_{i\in\mathcal{I}_{+}}u_i\log\lambda_i=0.
    \label{eq:residual-orthogonality}
\end{equation}
The residual is therefore orthogonal to both the constant vector and the log-eigenvalue vector on $\mathcal{I}_{+}$: it contains neither a common log-weight offset nor a linear trend in $\log\lambda_i$. Multiplying all estimated weights by a common positive constant changes only the fitted intercept and leaves the residual unchanged.

Different power-law priors rescale the same spectral coordinates.
We therefore examine whether this rescaling changes the residual
estimated from the margin contributions.

\begin{proposition}[Invariance to Power-Law Rescaling]
\label{prop:residual-invariance}
Fix the retained spectral basis, its eigenvalues, and the training triples. Let $\widetilde{c}_i(t)=\lambda_i^\alpha c_i(t)$ for $\alpha\in\mathbb{R}$, and compute the margin contributions directly from these rescaled coordinates. Then the positive support $\mathcal{I}_{+}$ and the residual $(u_i)_{i\in\mathcal{I}_{+}}$ remain unchanged, while the fitted exponent becomes $\gamma_{\mathrm e}+2\alpha$.
\end{proposition}

The rescaling shifts the estimated log weights by $-2\alpha\log\lambda_i$, which the fitted power-law trend absorbs entirely. 
The residual therefore remains unchanged; the full derivation is in Appendix~\ref{app:residual-invariance-proof}.\looseness=-1

\begin{algorithm}[t]
\caption{Relevance-Guided Residual Calibration (RRC)}
\label{alg:rrc}
\begin{algorithmic}[1]
\Require Document embeddings; training query embeddings and relevance
judgments; target dimension $k$; spectral prior rule.
\Ensure A fixed compression map $T_k$.

\State Compute the corpus mean $\boldsymbol{\mu}$, spectral basis $\mathbf{U}_k$, and eigenvalues $\{\lambda_i\}_{i=1}^{k}$ using Eq.~\eqref{eq:covariance}, and obtain the positive prior weights $\{p_i\}_{i=1}^{k}$ from the chosen compressor.

\State Construct training triples and compute their spectral coordinates using Eq.~\eqref{eq:coordinates}.

\State Estimate the query-balanced margin means $\widehat{m}_i$ and variances $\widehat{v}_i$ using Eq.~\eqref{eq:est-moments}.

\State Compute $\widehat{e}_i\gets(\widehat{m}_i)_{+}/\widehat{v}_i$ and perform \emph{support selection}: $\mathcal{I}_{+}\gets\{i\leq k:\widehat{e}_i>0\}$.

\State Fit the log-linear spectral trend on $\mathcal{I}_{+}$ and compute $u_i$ using Eqs.~\eqref{eq:powerlaw-fit} and~\eqref{eq:residual}.

\State Apply \emph{residual weighting}: set $w_i\gets p_i e^{u_i}$ for $i\in\mathcal{I}_{+}$, and $w_i\gets0$ otherwise.

\State \Return
$T_k:\mathbf{x}\mapsto
\operatorname{diag}(\sqrt{w_1},\ldots,\sqrt{w_k})
\mathbf{U}_k^{\top}(\mathbf{x}-\boldsymbol{\mu})$.
\end{algorithmic}
\end{algorithm}

\subsubsection{Calibrating a Spectral Prior}
\label{sec:prior-calibration}

We calibrate the chosen spectral prior by adding the residual from Eq.~\eqref{eq:residual} to its log weights, as defined in Eq.~\eqref{eq:calibration}.
Residual weighting thus multiplies each prior weight in $\mathcal{I}_{+}$ by $e^{u_i}$, increasing it when $u_i>0$ and decreasing it when $u_i<0$. Support selection assigns zero weight to directions outside $\mathcal{I}_{+}$.
Algorithm~\ref{alg:rrc} summarises the complete calibration procedure. \looseness=-1

For a fixed spectral basis and training triples, we compute the residual once in the coordinates of Eq.~\eqref{eq:coordinates} and reuse it across positive spectral priors, including all choices in Eq.~\eqref{eq:priors}. 
For a power-law prior, the coordinate scaling $c_i(t)\mapsto\lambda_i^{-\gamma/2}c_i(t)$ is the case $\alpha=-\gamma/2$ of Proposition~\ref{prop:residual-invariance}. Computing margin contributions directly after this scaling therefore yields the same residual.

The resulting calibration has two useful properties. First, for a power-law prior $p_i\propto\lambda_i^{-\gamma_{\mathrm p}}$, Eq.~\eqref{eq:residual-orthogonality} ensures that fitting a power law to the calibrated weights on $\mathcal{I}_{+}$ recovers $\gamma_{\mathrm p}$. We therefore adjust individual directions while preserving the prior's fitted exponent on this support. Second, if the estimated margin-based weights follow a power law over all of the leading $k$ principal directions, the residual is zero and calibration recovers the prior. Section~\ref{sec:exp-ablation} compares retaining the prior's trend with directly using the estimated margin-based weights.\looseness=-1

RRC reuses the compressor's spectral decomposition and is estimated once from training queries. Given training embeddings in the retained spectral coordinates, computing the margin means and variances costs $O\!\left(k\sum_q n_q\right)$, and trend fitting and removal cost $O(k)$. 
Folding the calibrated weights into Eq.~\eqref{eq:representation} yields a fixed map for document indexing and unseen-query compression, adding no per-vector computation beyond the underlying spectral compressor and requiring no test-time relevance feedback.
Given the spectral prior and training triples, RRC introduces no additional hyperparameters. \looseness=-1

\section{Experiments}
We examine whether RRC improves spectral compression (Section~\ref{sec:exp-main}), how its components contribute (Section~\ref{sec:exp-ablation}), and whether its gains persist after quantization (Section~\ref{sec:exp-quantization}).

\subsection{Experimental Setup}

\subsubsection{Datasets}

We evaluate on nine datasets spanning conventional and reasoning-intensive retrieval. The conventional group comprises MS MARCO Passage Ranking~\cite{bajaj2016ms} for web search and three datasets from BEIR~\cite{Thakur2021BEIR}: Natural Questions (\textbf{NQ})~\cite{KwiatkowskiPRCP19_NQ} for open-domain question answering, \textbf{HotpotQA}~\cite{Yang0ZBCSM18_HotpotQA} for multi-hop question answering, and \textbf{FiQA}~\cite{MaiaHFDMZB18_fiqa} for financial retrieval. For brevity, we refer to these four conventional retrieval datasets collectively as the \emph{\textbf{BEIR} group} throughout the experiments.

For reasoning-intensive retrieval, we select five datasets that cover the three domains of BRIGHT~\cite{su2024bright}: StackExchange (\textbf{Biology}, \textbf{Earth Science}, and \textbf{Economics}), coding (\textbf{LeetCode}), and theorem-based retrieval (\textbf{TheoremQA-T}). StackExchange contributes three datasets because it accounts for seven of BRIGHT's twelve.
We use the MS MARCO development queries and the test queries of the remaining
datasets.

\subsubsection{Retrieval Models}
\label{sec:models}

We use three open-source LLM-based dense retrievers: \textbf{Qwen3-8B}~\cite{zhang2025qwen3embedding}\footnote{\url{https://huggingface.co/Qwen/Qwen3-Embedding-8B}}, \textbf{GTE-7B}~\cite{li2023gte_7b}\footnote{\url{https://huggingface.co/Alibaba-NLP/gte-Qwen2-7B-instruct}}, and \textbf{DIVER-4B}~\cite{long2025diver}\footnote{\url{https://huggingface.co/AQ-MedAI/Diver-Retriever-4B}} (see Table~\ref{tab:models}). They differ in embedding dimension, support for Matryoshka Representation Learning (MRL), and whether they are trained for reasoning-intensive retrieval, covering a range of encoder properties relevant to compression.

\begin{table}[t]
\centering

\caption{Retrieval model statistics (Params: total parameters; MRL: Matryoshka Representation Learning support).}
\label{tab:models}
\setlength{\tabcolsep}{2pt}
\renewcommand\arraystretch{0.9}
\begin{tabular}{lcccccc}
\hline
\textbf{Model} & \textbf{Params} & \textbf{Dim} &\textbf{Length} & \textbf{MRL} & \textbf{Reasoning} & \textbf{Release} \\ \hline
Qwen3-8B       & 8.0B            & 4096    &32k      & \checkmark   & \ding{55} & Jun 2025         \\
GTE-7B         & 7.0B            & 3584   &32k  & \ding{55}    & \ding{55}    & Jun 2024         \\
DIVER-4B          & 4.0B           & 2560   &40k      & \ding{55}  & \checkmark  & Sep 2025         \\ 

\hline
\end{tabular}
\begin{flushleft}
\end{flushleft}
\vspace{-0.4cm}
\end{table}

\subsubsection{Baselines}
\label{sec:baselines}

We compare with post-hoc compression methods that keep the embedding encoder fixed, using the \textbf{Full-Dimensional Embeddings} as a reference. \textbf{Prefix Truncation} retains the first $k$ coordinates, as is standard in MRL-compatible models. \textbf{Random Truncation} uniformly selects $k$ coordinates without replacement~\cite{takeshita-etal-2025-randomly}, while \textbf{Random Projection} applies a Gaussian projection to $k$ dimensions.
For spectral methods, we evaluate \textbf{PCA} ($\gamma=0$), \textbf{Whitening} ($\gamma=1$), and \textbf{SpecTemp}~\cite{LiEK26_SpecTemp}, which selects the exponent from the corpus spectrum and target dimension.
We also include \textbf{Matryoshka-Adaptor}~\cite{yoon-etal-2024-matryoshka} in its unsupervised (\textbf{Matryoshka-Adaptor-U}) and supervised (\textbf{Matryoshka-Adaptor-S}) variants, which learn an adaptor over fixed embeddings at multiple output dimensions.

We apply RRC to each spectral method, yielding \textbf{PCA\,+\,RRC}, \textbf{Whitening\,+\,RRC}, and \textbf{SpecTemp\,+\,RRC}. 
Each calibrated method shares the spectral basis of its underlying compressor and uses relevance judgments to select active directions and calibrate their weights.
We also evaluate two full-dimensional references: PCA ($k{=}D$), which centres and rotates embeddings without truncation, and PCA ($k{=}D$)\,+\,RRC, which applies RRC within the full PCA basis.
Quantization methods and their combinations with spectral compression are evaluated separately in Section~\ref{sec:exp-quantization}.

\subsubsection{Evaluation Protocol}
\label{sec:setup-protocol}

We evaluate compression methods at target dimensions $k\in\{768,512,256,128\}$. We report \textbf{MRR@10} for MS MARCO and \textbf{nDCG@10} for the remaining datasets, with scores multiplied by 100. Retrieval uses cosine similarity: we $\ell_2$-normalise the original embeddings before compression and normalise the compressed vectors before scoring.
 
As not all datasets have a training split, RRC and supervised Matryoshka-Adaptor use five-fold cross-fitting over the evaluation queries with shared query partitions.
For each fold, we fit the transformation using judgments from the other four folds and apply it to both the held-out queries and the document corpus. We pool the held-out query results across folds to compute the final score. Each BRIGHT training fold contains 60--114 judged queries. 
The spectral basis and unsupervised adaptor are estimated from corpus embeddings without relevance labels.
We repeat experiments with seeds 1999, 5, and 2026 and average scores over seeds for each encoder--dataset configuration. 
We test per-dataset differences with a paired $t$-test over queries and aggregate differences across configurations with a Wilcoxon signed-rank test, which does not assume normality for the small number of heterogeneous configuration-level differences.
Aggregate results give equal weight to these configurations. Implementation details, including spectral estimation, negative sampling, and adaptor training, are provided in Appendix~\ref{app:implementation-details}. \looseness=-1

\subsection{Main Results}
\label{sec:exp-main}
We first assess whether RRC improves existing spectral compressors, then examine how calibration gains relate to the power-law fit of the label-estimated weights.

\begin{table*}[t]

\caption{Mean effectiveness ($\times100$) on the BEIR group (12 encoder--dataset configurations; MRR@10 on MS MARCO, nDCG@10 otherwise) and BRIGHT (15 configurations; nDCG@10). Group means and Win/Tie/Loss counts use configuration scores averaged over three seeds and rounded to two decimal places; ties indicate equality at that precision. Gray cells show full-dimensional references; bold marks the best compressed result in each column. Win/Tie/Loss compares each calibrated method with its corresponding uncalibrated prior. Counts cover 12 BEIR and 15 BRIGHT configurations for PCA ($k{=}D$)\,+\,RRC, and are summed over the four target dimensions for compressed methods (48 and 60 comparisons, respectively).}

\label{tab:main}
\setlength{\tabcolsep}{6pt}
\renewcommand\arraystretch{0.8}
\begin{tabular}{l cccc c cccc c}
\toprule
\multirow[b]{2}{*}{\raisebox{-0.8ex}{\makecell[c]{\makebox[3.2cm][c]{{Method} $\downarrow$ \hfill $k$ $\rightarrow$}}}}  & \multicolumn{5}{c}{BEIR group} & \multicolumn{5}{c}{BRIGHT} \\
\cmidrule(lr){2-6}\cmidrule(lr){7-11}
& 768 & 512 & 256 & 128 & Win/Tie/Loss & 768 & 512 & 256 & 128 & Win/Tie/Loss \\
\midrule
Full Dimension
  & \multicolumn{4}{c}{\cellcolor{gray!20}{56.52}} &
  & \multicolumn{4}{c}{\cellcolor{gray!20}{31.99}} & \\
  PCA ($k{=}D$)
  & \multicolumn{4}{c}{\cellcolor{gray!20}56.16} & --
  & \multicolumn{4}{c}{\cellcolor{gray!20}30.89} & -- \\
PCA ($k{=}D$) + RRC
  & \multicolumn{4}{c}{\cellcolor{gray!20}56.43} & 8/1/3
  & \multicolumn{4}{c}{\cellcolor{gray!20}{33.13}} & 10/0/5 \\
\midrule
Prefix Truncation    & 54.31 & 52.40 & 46.78 & 40.05 & --
                     & 30.94 & 29.62 & 25.53 & 19.81 & -- \\
Random Truncation    & 54.58 & 53.34 & 49.79 & 42.69 & --
                     & 30.00 & 29.05 & 26.09 & 21.02 & -- \\
Random Projection    & 54.44 & 53.39 & 50.18 & 43.80 & --
                     & 30.49 & 29.07 & 26.48 & 21.99 & -- \\
Matryoshka-Adaptor-U & 54.77 & 53.49 & 50.29 & 45.03 & --
                     & 30.68 & 29.74 & 26.93 & 23.00 & -- \\
Matryoshka-Adaptor-S & 55.53 & 54.48 & 51.90 & 47.38 & --
                     & 29.26 & 28.74 & 27.02 & 24.96 & -- \\
\midrule
PCA                  & 55.35 & 54.42 & 51.54 & 46.43 & --
                     & 29.91 & 29.19 & 27.38 & 24.34 & -- \\
PCA + RRC            & 55.63 & 54.64 & 51.80 & 46.66 & 28/4/16
                     & 32.20 & 31.49 & 28.97 & 25.21 & 37/0/23 \\
\midrule
Whitening            & 53.34 & 53.58 & 51.93 & 47.47 & --
                     & 29.34 & 29.72 & 28.28 & 25.33 & -- \\
Whitening + RRC      & 53.40 & 53.64 & 52.09 & 47.65 & 24/1/23
                     & 30.44 & 30.44 & 28.68 & 25.57 & 38/0/22 \\
\midrule
SpecTemp             & 55.66 & 54.95 & 52.48 & 47.71 & --
                     & 30.19 & 29.75 & 28.34 & 25.61 & -- \\
SpecTemp + RRC       & \textbf{55.87} & \textbf{55.14} & \textbf{52.76} & \textbf{48.00} & 29/2/17
                     & \textbf{32.53} & \textbf{31.96} & \textbf{29.99} & \textbf{26.66} & 41/0/19 \\
\bottomrule
\end{tabular}
\end{table*}

\subsubsection{Overall Comparison}   
Table~\ref{tab:main} summarises mean effectiveness on the BEIR group and BRIGHT, including full-dimensional references; Tables~\ref{tab:main_results_beir} and~\ref{tab:main_results_bright} in the appendix report per-dataset results.
Adding RRC improves the mean effectiveness of PCA, whitening, and SpecTemp at all four target dimensions in both groups, and SpecTemp\,+\,RRC achieves the highest mean among the compression methods in every column.
At $k=768$, it is 0.65 points below the original embeddings on the BEIR group and 0.54 points above them on BRIGHT.

The full-dimensional references separate centring, calibration, and truncation. Centring alone lowers both group means, whereas RRC improves the centred representations. For PCA, the mean calibration gains at $k=768$ are comparable to those at $k=D$, showing that these gains are retained after truncation in both groups. At $k=768$, SpecTemp\,+\,RRC comes within 0.6 points of PCA ($k{=}D$)\,+\,RRC in both groups while using substantially fewer coordinates. On BRIGHT, uncalibrated SpecTemp remains below the original embeddings; adding RRC raises its mean effectiveness above them.\looseness=-1

Matryoshka-Adaptor-S achieves substantial gains in some configurations, including Qwen3 on Biology, but its effectiveness varies widely in the label-scarce BRIGHT setting.
On the BEIR group, it outperforms its unsupervised counterpart in mean effectiveness at every target dimension, but SpecTemp\,+\,RRC still achieves a higher group mean at every target dimension.
Prefix truncation remains strong on DIVER--Biology, where it outperforms SpecTemp\,+\,RRC at three of the four target dimensions. 
A relevant distinction is that RRC rescales centred spectral coordinates, whereas prefix truncation retains the first $k$ coordinates of the original embeddings without centring.

Calibration gains are larger on BRIGHT, where SpecTemp\,+\,RRC improves over SpecTemp by 2.34 points on average at $k=768$. Changes are generally smaller on the BEIR group. We next examine how this variation relates to the structure of the label-estimated weights.

 \begin{figure}[!t]
  \centering
  \includegraphics[width=0.9273\columnwidth]{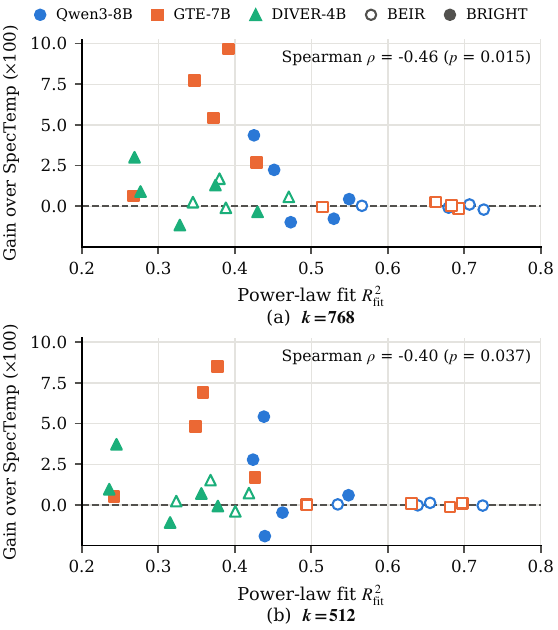}
    \caption{Power-law fit quality versus the gain of SpecTemp\,+\,RRC over SpecTemp at $k=768$ and $512$. Fit scores use 80 queries per training fold on the BEIR group and 60--114 on BRIGHT; 
    gains use the main experiment's training budgets. Each point averages three seeds. Colours and shapes denote encoders; filled markers denote BRIGHT and hollow markers the BEIR group. Gains use nDCG@10 ($\times100$), except MRR@10 on MS MARCO.
    }
  \label{fig:r2-gain}
  \vspace{-0.3cm}
\end{figure}

\subsubsection{When Does Calibration Help?}
\label{sec:exp-when}

To understand variation in calibration gains across configurations, we measure how closely the estimated margin-based weights follow a power law. For each training fold, we compute the coefficient of determination for the fit in Eq.~\eqref{eq:powerlaw-fit}:
\begin{equation}
  R^2_{\mathrm{fit}}
  =1-\frac{\sum_{i\in\mathcal I_+}u_i^2}
              {\sum_{i\in\mathcal I_+}
                (\log\widehat e_i-\bar\ell)^2},
  \label{eq:r2-fit}
\end{equation}
where $\bar\ell=|\mathcal I_+|^{-1} \sum_{i\in\mathcal I_+}\log\widehat e_i$, and $u_i$ is the residual defined in Eq.~\eqref{eq:residual}. A high value means that the fitted spectral trend explains most of the variation in the estimated log weights. This diagnostic measures agreement with the power-law family, whose exponent is fitted freely.
Because sample size affects estimation noise and $R^2_{\mathrm{fit}}$, we use 80 training queries per fold on the BEIR group to match the 60--114 queries on BRIGHT.
Fit scores are averaged over folds and seeds. We examine their association with calibration gains from the main experiments, which retain their original training budgets.

Figure~\ref{fig:r2-gain} shows a negative association between fit quality and calibration gain at both dimensions. Configurations with high fit scores cluster near zero gain, whereas those with lower scores include the largest improvements. The pattern also varies across encoders: GTE and Qwen3 occupy the high-fit cluster on the conventional benchmarks, while DIVER has lower fit scores on these same datasets and benefits from calibration on NQ.
Within BRIGHT, calibration gains also differ among configurations with similar fit scores.
Overall, larger calibration gains tend to occur where a single power-law trend captures less of the variation in the estimated margin-based weights.

\subsection{Ablation Study}
\label{sec:exp-ablation}

RRC combines a spectral prior with residual corrections estimated from relevance judgments. 
We evaluate calibration across fixed prior exponents and disentangle the effects of support selection, residual weighting, and spectral trend retention.
We use BRIGHT, where the main experiments show larger calibration gains, using all 15 encoder--dataset configurations at $k\in\{768,512\}$ and averaging over three seeds.

\begin{figure}[t]
  \centering
  \includegraphics[width=\columnwidth]{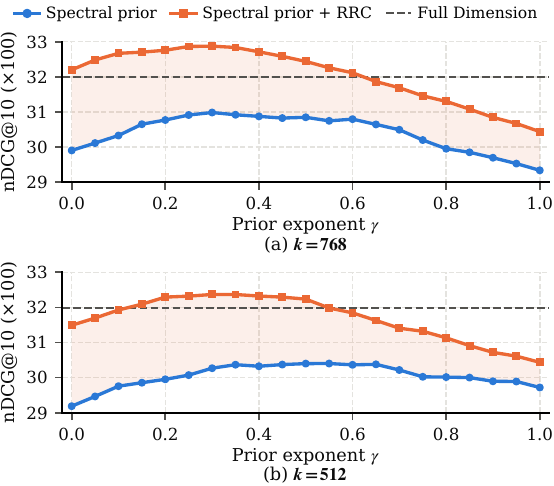}
  \caption{Mean nDCG@10 ($\times100$) over 15 encoder--dataset configurations on BRIGHT, averaged over three seeds, for the power-law prior $\lambda_i^{-\gamma}$ and the same prior with RRC, at $k=768$ and $512$. Shaded regions show the effectiveness gains from RRC; the dashed line is the mean full-dimensional score of the same configurations.\looseness=-1}
  \label{fig:gamma-sweep}
\end{figure}

\subsubsection{Calibration across the Spectral Prior Family.}

To evaluate whether RRC is effective across different spectral transformations, we apply it to the family of power-law spectral priors $p_i\propto\lambda_i^{-\gamma}$ over $\gamma\in \{0,0.05,\ldots,1\}$ at $k=768$ and $512$. This covers PCA at $\gamma=0$, whitening at $\gamma=1$, and the spectral transformations between them.

Figure~\ref{fig:gamma-sweep} shows higher mean nDCG@10 with RRC at every tested exponent and both dimensions, with generally larger gains at smaller $\gamma$. At intermediate exponents, the calibrated priors also exceed full-dimensional effectiveness on average, while the uncalibrated curve remains below it throughout the sweep. Calibration therefore provides additional gains when the prior exponent is held fixed, across the tested range of spectral weighting rules.

\begin{table}[t]
  \centering
\caption{Sources of the calibration gain on BRIGHT, averaged over 15 encoder--dataset configurations and 3 seeds (nDCG@10 $\times100$). SpecTemp on $\mathcal{I}_{+}$ applies only support selection. $\Delta$ and Wins denote the difference between the reported mean scores and the number of outperforming configurations relative to the listed reference; $^{*}$ marks $p<0.05$ (Wilcoxon signed-rank test over configurations). All label-based variants share cross-fitting folds, $\widehat{\mathbf e}$, and $\mathcal{I}_{+}$.}
  \label{tab:weight-source}
  \footnotesize
  \setlength{\tabcolsep}{2pt}
\begin{tabular}{llcccccc}
\toprule
\multirow{2}{*}{Reference}
  & \multirow{2}{*}{Weights}
  & \multicolumn{3}{c}{$k=768$}
  & \multicolumn{3}{c}{$k=512$}\\
\cmidrule(lr){3-5}\cmidrule(lr){6-8}
  & & nDCG & $\Delta$ & Wins
    & nDCG & $\Delta$ & Wins\\
\midrule
  --
    & SpecTemp
    & 30.19 & -- & --
    & 29.75 & -- & --\\
  \midrule
\multirow{3}{*}{SpecTemp}
  & SpecTemp on $\mathcal{I}_{+}$
  & 31.73 & $+1.54^{*}$ & 12/15 & 31.30 & $+1.55^{*}$ & 10/15\\
    & Direct $\widehat{\mathbf e}$
    & 31.50 & $+1.31$ & 7/15
    & 31.24 & $+1.49$ & 8/15\\
  & SpecTemp\,+\,RRC
  & \textbf{32.53} & $+2.34^{*}$ & 11/15 & \textbf{31.96} & $+2.21^{*}$ & 11/15\\
  \midrule
SpecTemp on $\mathcal{I}_{+}$
  & SpecTemp\,+\,RRC
  & -- & $+0.80$ & 12/15 & -- & $+0.66$ & 11/15\\
  \midrule
Direct $\widehat{\mathbf e}$
    & SpecTemp\,+\,RRC
    & -- & $+1.03$ & 11/15
    & -- & $+0.72$ & 7/15\\
  \bottomrule
\end{tabular}
  \vspace{-0.3cm}
\end{table}

\subsubsection{Sources of the Calibration Gain}
RRC applies two steps to the spectral prior: \emph{support selection}, which sets the weights outside $\mathcal{I}_{+}$ to zero, and \emph{residual weighting}, which multiplies the remaining prior weights by $e^{u_i}$ (Eq.~\eqref{eq:calibration}).
Using SpecTemp as the spectral prior, we evaluate \textbf{SpecTemp on $\mathcal I_{+}$}, which applies only support selection.  We also include \textbf{Direct $\widehat{\mathbf e}$}, which uses the estimated margin-based weights $\widehat{\mathbf e}$ from Eq.~\eqref{eq:empirical-weights-hat} as the final weights. Direct $\widehat{\mathbf e}$ shares RRC's support and residuals but retains the label-fitted exponent $\gamma_{\mathrm e}$.

Table~\ref{tab:weight-source} shows that both steps increase mean effectiveness at both target dimensions. 
Both steps use the same margin-based estimates: positive estimates define $\mathcal{I}_{+}$, and their log-space residuals reweight the prior on this support. 
Support selection alone significantly improves over SpecTemp, and residual weighting further increases the group mean on the same support in most configurations. The full RRC, which combines both steps, achieves the highest mean and significantly improves over SpecTemp.

The comparison with Direct $\widehat{\mathbf e}$ isolates the source of the spectral trend while holding the support and residuals fixed. For a power-law prior $p_i\propto\lambda_i^{-\gamma}$, the two weightings differ only in their exponents: substituting Eq.~\eqref{eq:residual} into Eq.~\eqref{eq:calibration} gives
\begin{equation}
  \log w_i=\log\widehat{e}_i+(\gamma_{\mathrm{e}}-\gamma)\log\lambda_i
  +\text{const},\qquad i\in\mathcal{I}_{+},
  \label{eq:direct-identity}
\end{equation}
where $\text{const}$ is an additive constant independent of $i$ and does not affect rankings. 
Thus, using $\widehat{\mathbf{e}}$ directly is equivalent to RRC with the prior exponent set to $\gamma_{\mathrm e}$. This exponent minimises the log-weight fitting error in Eq.~\eqref{eq:powerlaw-fit}, but does not necessarily maximise retrieval effectiveness. RRC instead uses the exponent of the chosen spectral prior while retaining the same residual corrections. 
With the support and residual held fixed, using the SpecTemp exponent yields higher mean retrieval scores than retaining $\gamma_{\mathrm e}$ in this experiment.

\subsection{Quantization and Storage Efficiency}
\label{sec:exp-quantization}
The preceding experiments compare compression methods at a fixed output dimension. We now examine their storage implications through two questions: (i) do the gains from RRC persist after quantization? (ii) does calibrated spectral compression offer a better storage--effectiveness trade-off than directly quantizing the original embeddings?
We use SpecTemp as the spectral prior and evaluate all three encoders on the five BRIGHT datasets, retaining $k=768$ dimensions for the spectral methods. RRC follows the same query-level cross-fitting protocol as the main experiments, and quantizers are trained on corpus embeddings only. We report nDCG@10 ($\times100$), averaged over three seeds and then over the 15 encoder--dataset configurations.

\subsubsection{Calibration under Quantization}
\label{sec:exp-quantization-calibration}

We first compare SpecTemp with and without RRC under identical encoding configurations, including the original full-dimensional embeddings as a reference.
For each representation, we evaluate FP32 and FP16 storage, as well as 8- and 4-bit scalar quantization (SQ8 and SQ4) implemented in Faiss~\cite{douze2024faiss}. Scalar quantization encodes each coordinate independently using 256 levels for SQ8 and 16 levels for SQ4.
We also evaluate product quantization (PQ)~\cite{jegou2010product} and an OPQ rotation~\cite{Ge2013OPQ} followed by PQ, denoted OPQ below. 
Binary encoding is evaluated only on the full-dimensional embeddings, as sign binarisation discards positive coordinate scaling.
PQ and OPQ use $M$ subquantizers with 8 bits each, resulting in $M$ bytes per document.

Table~\ref{tab:quantization} shows that RRC improves mean effectiveness over SpecTemp under every encoding configuration evaluated for the spectral methods. FP16 and SQ8 largely preserve the gain observed with FP32, whereas SQ4 reduces it. The gains remain positive at every tested PQ and OPQ code length, showing that calibration remains beneficial across different quantization methods and compression levels.

\begin{table}[t]
  \centering
    \caption{Mean nDCG@10 ($\times100$) over 15 BRIGHT configurations and three seeds. Spectral methods retain $k=768$ dimensions. The upper block compares floating-point and scalar encodings (plus full-dimensional binary); the lower block compares PQ and OPQ at $M$ bytes per document.}
  \label{tab:quantization}
  \setlength{\tabcolsep}{2.5pt}
  \renewcommand{\arraystretch}{0.802}
  \begin{tabular}{@{}llccccc@{}}
    \toprule
    \multicolumn{2}{@{}l}{\textit{Without product quantization}}
      & FP32 & FP16 & SQ8 & SQ4 & Binary\\
    \midrule
    \multicolumn{2}{@{}l}{Full Dimension}   & 31.99 & 31.99 & 32.02 & \textbf{31.80} & 30.21\\
    \multicolumn{2}{@{}l}{SpecTemp}         & 30.19 & 30.18 & 30.17 & 30.00 & --\\
    \multicolumn{2}{@{}l}{SpecTemp\,+\,RRC} & \textbf{32.53} & \textbf{32.53} & \textbf{32.49} & 31.70 & --\\
    \midrule
    \midrule
        \multicolumn{2}{@{}l}{\textit{With product quantization}}
      & \multicolumn{5}{c}{Bytes per document ($M$)}\\
    \cmidrule(l){3-7}
    Representation & Codec & 16 & 32 & 64 & 128 & 256\\
    \midrule
    Full Dimension & PQ
      & 4.88 & 8.14 & 13.11 & 20.64 & 27.63\\
    Full Dimension & OPQ
      & 17.53 & 22.14 & 26.38 & 28.82 & 30.49\\
    SpecTemp & PQ
      & 13.34 & 19.06 & 24.37 & 27.45 & 29.64\\
    SpecTemp & OPQ
      & 17.70 & 22.37 & 26.32 & 28.80 & 29.98\\
    SpecTemp\,+\,RRC & PQ
      & 15.55 & 21.72 & 27.22 & 30.35 & 32.08\\
    SpecTemp\,+\,RRC & OPQ
      & \textbf{19.40} & \textbf{25.38} & \textbf{29.51}
      & \textbf{31.55} & \textbf{32.13}\\
      \bottomrule
  \end{tabular}
  \vspace{-0.3cm}
\end{table}

\subsubsection{Effectiveness at Matched Storage Budgets}
\label{sec:exp-storage-budget}

We compare the three representations across encodings using the empirical Pareto fronts in Figure~\ref{fig:storage-pareto}. Each operating point uses the same encoding configuration across all 15 encoder--dataset combinations. We average effectiveness and storage before selecting the fronts, measuring storage as mean bytes per document vector.

Up to 128 bytes per document, the SpecTemp frontier closely follows that of the original embeddings. At larger budgets, it levels off below the full-dimensional FP32 mean score. Increasing encoding precision alone therefore does not close the effectiveness gap. Adding RRC raises the spectral frontier, providing higher mean effectiveness at matched storage budgets.
Under the same SQ8 encoding, SpecTemp\,+\,RRC achieves a mean nDCG@10 of 32.49 at 768 bytes per document, compared with 32.02 at 3,413 bytes per document on average for the original embeddings. Thus, calibrated spectral compression achieves higher mean effectiveness while requiring substantially less storage per document.

 \begin{figure}[t]
  \centering
  \includegraphics[width=\columnwidth]{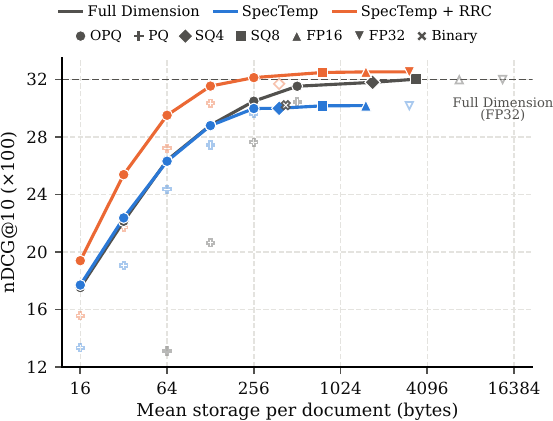}
  \caption{Storage--effectiveness trade-offs on BRIGHT, averaged over three seeds and 15 encoder--dataset configurations.
  Colours identify representations and marker shapes identify encodings; spectral methods retain $k=768$ dimensions. Lines connect the Pareto-optimal operating points of each representation. Dominated points are faded, except for the binary baseline, which is highlighted for visibility. Full-dimensional PQ at $M=16$ and $32$ falls below the plotted range. The dashed line marks mean full-dimensional FP32 effectiveness. 
  Storage excludes shared transformation matrices and codebooks, whose costs are amortised over the collection.}
  \label{fig:storage-pareto}
  \vspace{-0.2cm}
\end{figure}

\section{Conclusion}

We presented RRC, a closed-form calibration method for spectral compression in dense retrieval that requires no iterative model training or additional inference stage. RRC retains directions with positive estimated margin-based weights, separates their log weights into a power-law trend and a residual, and uses the residual to reweight the prior while preserving the chosen compressor's spectral trend. Across three encoders and nine datasets, RRC improves mean effectiveness for all tested spectral priors, especially on reasoning-intensive retrieval, and these gains persist under quantization.\looseness=-1

RRC requires relevance judgments from the target task, and its gains tend to be smaller when the label-estimated weights closely follow a power law. 
Future work will investigate alternatives to corpus-mean centring and use relevance evidence to learn the projection subspace itself.





\bibliographystyle{ACM-Reference-Format}
\bibliography{main}

@inproceedings{KarpukhinOMLWEC20_DPR,
  author       = {Vladimir Karpukhin and
                  Barlas Oguz and
                  Sewon Min and
                  Patrick S. H. Lewis and
                  Ledell Wu and
                  Sergey Edunov and
                  Danqi Chen and
                  Wen{-}tau Yih},
  title        = {Dense Passage Retrieval for Open-Domain Question Answering},
  booktitle    = {Proceedings of the 2020 Conference on Empirical Methods in Natural
                  Language Processing, {EMNLP} 2020},
  pages        = {6769--6781},
  publisher    = {Association for Computational Linguistics},
  year         = {2020},
  url          = {https://doi.org/10.18653/v1/2020.emnlp-main.550},
  doi          = {10.18653/V1/2020.EMNLP-MAIN.550},
  bibsource    = {dblp computer science bibliography, https://dblp.org}
}

@article{bajaj2016ms,
title={Ms marco: A human generated machine reading comprehension dataset},
author={Bajaj, Payal and Campos, Daniel and Craswell, Nick and Deng, Li and Gao, Jianfeng and Liu, Xiaodong and Majumder, Rangan and McNamara, Andrew and Mitra, Bhaskar and Nguyen, Tri and others},
journal={arXiv preprint arXiv:1611.09268},
year={2016}
}

@inproceedings{
thakur2021beir,
title={{BEIR}: A Heterogeneous Benchmark for Zero-shot Evaluation of Information Retrieval Models},
author={Nandan Thakur and Nils Reimers and Andreas R{\"u}ckl{\'e} and Abhishek Srivastava and Iryna Gurevych},
booktitle={Thirty-fifth Conference on Neural Information Processing Systems Datasets and Benchmarks Track (Round 2)},
year={2021},
url={https://openreview.net/forum?id=wCu6T5xFjeJ}
}

@article{KwiatkowskiPRCP19_NQ,
  author       = {Tom Kwiatkowski and
                  Jennimaria Palomaki and
                  Olivia Redfield and
                  Michael Collins and
                  Ankur P. Parikh and
                  Chris Alberti and
                  Danielle Epstein and
                  Illia Polosukhin and
                  Jacob Devlin and
                  Kenton Lee and
                  Kristina Toutanova and
                  Llion Jones and
                  Matthew Kelcey and
                  Ming{-}Wei Chang and
                  Andrew M. Dai and
                  Jakob Uszkoreit and
                  Quoc Le and
                  Slav Petrov},
  title        = {Natural Questions: a Benchmark for Question Answering Research},
  journal      = {Trans. Assoc. Comput. Linguistics},
  volume       = {7},
  pages        = {452--466},
  year         = {2019},
  url          = {https://doi.org/10.1162/tacl\_a\_00276},
  doi          = {10.1162/TACL\_A\_00276},
  bibsource    = {dblp computer science bibliography, https://dblp.org}
}

@inproceedings{LinALOLMY023_dragon,
  author       = {Sheng{-}Chieh Lin and
                  Akari Asai and
                  Minghan Li and
                  Barlas Oguz and
                  Jimmy Lin and
                  Yashar Mehdad and
                  Wen{-}tau Yih and
                  Xilun Chen},
  title        = {How to Train Your Dragon: Diverse Augmentation Towards Generalizable Dense Retrieval},
  booktitle    = {Findings of the Association for Computational Linguistics: {EMNLP}   2023, Singapore, December 6-10, 2023},
  pages        = {6385--6400},
  publisher    = {Association for Computational Linguistics},
  year         = {2023},
  url          = {https://doi.org/10.18653/v1/2023.findings-emnlp.423},
  doi          = {10.18653/V1/2023.FINDINGS-EMNLP.423},
  bibsource    = {dblp computer science bibliography, https://dblp.org}
}

@inproceedings{MaiaHFDMZB18_fiqa,
  author       = {Macedo Maia and
                  Siegfried Handschuh and
                  Andr{\'{e}} Freitas and
                  Brian Davis and
                  Ross McDermott and
                  Manel Zarrouk and
                  Alexandra Balahur},
  title        = {WWW'18 Open Challenge: Financial Opinion Mining and Question Answering},
  booktitle    = {Companion of the The Web Conference 2018 on The Web Conference 2018,
                  {WWW} 2018, Lyon , France, April 23-27, 2018},
  pages        = {1941--1942},
  publisher    = {{ACM}},
  year         = {2018},
  url          = {https://doi.org/10.1145/3184558.3192301},
  bibsource    = {dblp computer science bibliography, https://dblp.org}
}

@misc{zhang2025qwen3embedding,
      title={Qwen3 Embedding: Advancing Text Embedding and Reranking Through Foundation Models}, 
      author={Yanzhao Zhang and Mingxin Li and Dingkun Long and Xin Zhang and Huan Lin and Baosong Yang and Pengjun Xie and An Yang and Dayiheng Liu and Junyang Lin and Fei Huang and Jingren Zhou},
      year={2025},
      eprint={2506.05176},
      archivePrefix={arXiv},
      primaryClass={cs.CL},
      url={https://arxiv.org/abs/2506.05176}, 
}

@misc{long2025diver,
      title={DIVER: A Multi-Stage Approach for Reasoning-intensive Information Retrieval}, 
      author={Meixiu Long and Duolin Sun and Dan Yang and Junjie Wang and Yue Shen and Jian Wang and Peng Wei and Jinjie Gu and Jiahai Wang},
      year={2025},
      eprint={2508.07995},
      archivePrefix={arXiv},
      primaryClass={cs.IR},
      url={https://arxiv.org/abs/2508.07995}, 
}

@inproceedings{KusupatiBRWSRHC22_MRL,
  author       = {Aditya Kusupati and
                  Gantavya Bhatt and
                  Aniket Rege and
                  Matthew Wallingford and
                  Aditya Sinha and
                  Vivek Ramanujan and
                  William Howard{-}Snyder and
                  Kaifeng Chen and
                  Sham M. Kakade and
                  Prateek Jain and
                  Ali Farhadi},
  title        = {Matryoshka Representation Learning},
  booktitle    = {Advances in Neural Information Processing Systems 35: Annual Conference
                  on Neural Information Processing Systems 2022, NeurIPS 2022, New Orleans,
                  LA, USA, November 28 - December 9, 2022},
  year         = {2022},
  url          = {http://papers.nips.cc/paper\_files/paper/2022/hash/c32319f4868da7613d78af9993100e42-Abstract-Conference.html},
  bibsource    = {dblp computer science bibliography, https://dblp.org}
}

@misc{li2023gte_7b,
      title={Towards General Text Embeddings with Multi-stage Contrastive Learning}, 
      author={Zehan Li and Xin Zhang and Yanzhao Zhang and Dingkun Long and Pengjun Xie and Meishan Zhang},
      year={2023},
      eprint={2308.03281},
      archivePrefix={arXiv},
      primaryClass={cs.CL},
      url={https://arxiv.org/abs/2308.03281}, 
}

@inproceedings{ SatopaaAIR11_Kneedle,
  author       = {Ville Satopaa and
                  Jeannie R. Albrecht and
                  David E. Irwin and
                  Barath Raghavan},
  title        = {Finding a "Kneedle" in a Haystack: Detecting Knee Points in System
                  Behavior},
  booktitle    = {31st {IEEE} International Conference on Distributed Computing Systems
                  Workshops {(ICDCS} 2011 Workshops), 20-24 June 2011, Minneapolis,
                  Minnesota, {USA}},
  pages        = {166--171},
  publisher    = {{IEEE} Computer Society},
  year         = {2011},
  url          = {https://doi.org/10.1109/ICDCSW.2011.20},
  doi          = {10.1109/ICDCSW.2011.20},
  bibsource    = {dblp computer science bibliography, https://dblp.org}
}

@inproceedings{ma-etal-2021-simple,
    title = "Simple and Effective Unsupervised Redundancy Elimination to Compress Dense Vectors for Passage Retrieval",
    author = "Ma, Xueguang  and
      Li, Minghan  and
      Sun, Kai  and
      Xin, Ji  and
      Lin, Jimmy",
    booktitle = "Proceedings of the 2021 Conference on Empirical Methods in Natural Language Processing",
    month = nov,
    year = "2021",
    address = "Online and Punta Cana, Dominican Republic",
    publisher = "Association for Computational Linguistics",
    url = "https://aclanthology.org/2021.emnlp-main.227/",
    pages = "2854--2859",
}

@misc{su2021whiteningsentence,
      title={Whitening Sentence Representations for Better Semantics and Faster Retrieval}, 
      author={Jianlin Su and Jiarun Cao and Weijie Liu and Yangyiwen Ou},
      year={2021},
      eprint={2103.15316},
      archivePrefix={arXiv},
      primaryClass={cs.CL},
      url={https://arxiv.org/abs/2103.15316}, 
}

@inproceedings{huang-etal-2021-whiteningbert-easy,
    title = "{W}hitening{BERT}: An Easy Unsupervised Sentence Embedding Approach",
    author = "Huang, Junjie  and
      Tang, Duyu  and
      Zhong, Wanjun  and
      Lu, Shuai  and
      Shou, Linjun  and
      Gong, Ming  and
      Jiang, Daxin  and
      Duan, Nan",
    booktitle = "Findings of the Association for Computational Linguistics: EMNLP 2021",
    month = nov,
    year = "2021",
    address = "Punta Cana, Dominican Republic",
    publisher = "Association for Computational Linguistics",
    url = "https://aclanthology.org/2021.findings-emnlp.23/",
    pages = "238--244"
}

@article{johnson1984extensions,
  title={Extensions of Lipschitz mappings into a Hilbert space},
  author={Johnson, William B and Lindenstrauss, Joram and others},
  journal={Contemporary mathematics},
  volume={26},
  number={189-206},
  pages={1},
  year={1984},
    url={https://api.semanticscholar.org/CorpusID:117819162}
}

@inproceedings{raunak-etal-2019-effective,
    title = "Effective Dimensionality Reduction for Word Embeddings",
    author = "Raunak, Vikas  and
      Gupta, Vivek  and
      Metze, Florian",
    booktitle = "Proceedings of the 4th Workshop on Representation Learning for NLP (RepL4NLP-2019)",
    month = aug,
    year = "2019",
    address = "Florence, Italy",
    publisher = "Association for Computational Linguistics",
    url = "https://aclanthology.org/W19-4328/",
    pages = "235--243",
}

@article{jegou2010product,
  author       = {Herv{\'{e}} J{\'{e}}gou and
                  Matthijs Douze and
                  Cordelia Schmid},
  title        = {Product Quantization for Nearest Neighbor Search},
  journal      = {{IEEE} Trans. Pattern Anal. Mach. Intell.},
  volume       = {33},
  number       = {1},
  pages        = {117--128},
  year         = {2011},
  url          = {https://doi.org/10.1109/TPAMI.2010.57},
  doi          = {10.1109/TPAMI.2010.57},
  bibsource    = {dblp computer science bibliography, https://dblp.org}
}

@online{kexuefm-9079,
  title   = {When BERT Whitening Introduces Hyperparameters: There Is Always One That Suits You},
  author  = {Su, Jianlin},
  year    = {2022},
  month   = {May},
  url     = {https://kexue.fm/archives/9079},
  note    = {Chinese blog post}
}

@inproceedings{zuo2026efficiencyembeddingcompress,
    title = "More Than Efficiency: Embedding Compression Improves Domain Adaptation in Dense Retrieval",
    author = "Zuo, Chunsheng  and
      Khashabi, Daniel",
    booktitle = "Proceedings of the First Workshop on Structured Understanding, Retrieval, and Generation in the {LLM} Era ({SURG}e{LLM} 2026)",
    month = jul,
    year = "2026",
    address = "San Diego, California, United States",
    publisher = "Association for Computational Linguistics",
    url = "https://aclanthology.org/2026.surgellm-1.24/",
    doi = "10.18653/v1/2026.surgellm-1.24",
    pages = "361--377",
    ISBN = "979-8-89176-406-4"
}

@inproceedings{Yang0ZBCSM18_HotpotQA,
  author       = {Zhilin Yang and
                  Peng Qi and
                  Saizheng Zhang and
                  Yoshua Bengio and
                  William W. Cohen and
                  Ruslan Salakhutdinov and
                  Christopher D. Manning},
  title        = {HotpotQA: {A} Dataset for Diverse, Explainable Multi-hop Question
                  Answering},
  booktitle    = {Proceedings of the 2018 Conference on Empirical Methods in Natural
                  Language Processing, Brussels, Belgium, October 31 - November 4, 2018},
  pages        = {2369--2380},
  publisher    = {Association for Computational Linguistics},
  year         = {2018},
  url          = {https://doi.org/10.18653/v1/d18-1259},
  doi          = {10.18653/V1/D18-1259},
  bibsource    = {dblp computer science bibliography, https://dblp.org}
}

@inproceedings{liu-etal-2022-dimension,
    title = "Dimension Reduction for Efficient Dense Retrieval via Conditional Autoencoder",
    author = "Liu, Zhenghao  and
      Zhang, Han  and
      Xiong, Chenyan  and
      Liu, Zhiyuan  and
      Gu, Yu  and
      Li, Xiaohua",
    booktitle = "Proceedings of the 2022 Conference on Empirical Methods in Natural Language Processing",
    month = dec,
    year = "2022",
    address = "Abu Dhabi, United Arab Emirates",
    publisher = "Association for Computational Linguistics",
    url = "https://aclanthology.org/2022.emnlp-main.384/",
    doi = "10.18653/v1/2022.emnlp-main.384",
    pages = "5692--5698"
}

@article{douze2024faiss,
      title={The Faiss library},
      author={Matthijs Douze and Alexandr Guzhva and Chengqi Deng and Jeff Johnson and Gergely Szilvasy and Pierre-Emmanuel Mazaré and Maria Lomeli and Lucas Hosseini and Hervé Jégou},
      year={2024},
      eprint={2401.08281},
      archivePrefix={arXiv},
      primaryClass={cs.LG}
}

@inproceedings{zhang-etal-2024-evaluating-unsupervised,
    title = "Evaluating Unsupervised Dimensionality Reduction Methods for Pretrained Sentence Embeddings",
    author = "Zhang, Gaifan  and
      Zhou, Yi  and
      Bollegala, Danushka",
    booktitle = "Proceedings of the 2024 Joint International Conference on Computational Linguistics, Language Resources and Evaluation (LREC-COLING 2024)",
    month = may,
    year = "2024",
    address = "Torino, Italia",
    publisher = "ELRA and ICCL",
    url = "https://aclanthology.org/2024.lrec-main.579/",
    pages = "6530--6543"
}

@inproceedings{LiEK26_SpecTemp,
  author       = {Yongkang Li and
                  Panagiotis Eustratiadis and
                  Evangelos Kanoulas},
  title        = {Spectral Tempering for Embedding Compression in Dense Passage Retrieval},
  booktitle    = {Proceedings of the 49th International {ACM} {SIGIR} Conference on
                  Research and Development in Information Retrieval, {SIGIR} 2026, MelbourneVICAustralia,
                  July 20-24, 2026},
  pages        = {3941--3946},
  publisher    = {{ACM}},
  year         = {2026},
  url          = {https://doi.org/10.1145/3805712.3809896},
  doi          = {10.1145/3805712.3809896},
  bibsource    = {dblp computer science bibliography, https://dblp.org}
}

@inproceedings{su2024bright,
  author       = {Hongjin Su and
                  Howard Yen and
                  Mengzhou Xia and
                  Weijia Shi and
                  Niklas Muennighoff and
                  Han{-}yu Wang and
                  Haisu Liu and
                  Quan Shi and
                  Zachary S. Siegel and
                  Michael Tang and
                  Ruoxi Sun and
                  Jinsung Yoon and
                  Sercan {\"{O}}. Arik and
                  Danqi Chen and
                  Tao Yu},
  title        = {{BRIGHT:} {A} Realistic and Challenging Benchmark for Reasoning-Intensive
                  Retrieval},
  booktitle    = {The Thirteenth International Conference on Learning Representations,
                  {ICLR} 2025, Singapore, April 24-28, 2025},
  publisher    = {OpenReview.net},
  year         = {2025},
  url          = {https://openreview.net/forum?id=ykuc5q381b},
  bibsource    = {dblp computer science bibliography, https://dblp.org}
}

@inproceedings{yoon-etal-2024-search,
    title = "Search-Adaptor: Embedding Customization for Information Retrieval",
    author = "Yoon, Jinsung  and
      Chen, Yanfei  and
      Arik, Sercan  and
      Pfister, Tomas",
    booktitle = "Proceedings of the 62nd Annual Meeting of the Association for Computational Linguistics (Volume 1: Long Papers)",
    month = aug,
    year = "2024",
    address = "Bangkok, Thailand",
    publisher = "Association for Computational Linguistics",
    url = "https://aclanthology.org/2024.acl-long.661/",
    doi = "10.18653/v1/2024.acl-long.661",
    pages = "12230--12247"
}

@inproceedings{yoon-etal-2024-matryoshka,
    title = "Matryoshka-Adaptor: Unsupervised and Supervised Tuning for Smaller Embedding Dimensions",
    author = "Yoon, Jinsung  and
      Sinha, Rajarishi  and
      Arik, Sercan O  and
      Pfister, Tomas",
    booktitle = "Proceedings of the 2024 Conference on Empirical Methods in Natural Language Processing",
    month = nov,
    year = "2024",
    address = "Miami, Florida, USA",
    publisher = "Association for Computational Linguistics",
    url = "https://aclanthology.org/2024.emnlp-main.576/",
    doi = "10.18653/v1/2024.emnlp-main.576",
    pages = "10318--10336"
}

@inproceedings{zhang-etal-2025-smec,
    title = "{SMEC}:Rethinking Matryoshka Representation Learning for Retrieval Embedding Compression",
    author = "Zhang, Biao  and
      Chen, Lixin  and
      Liu, Tong  and
      Zheng, Bo",
    booktitle = "Proceedings of the 2025 Conference on Empirical Methods in Natural Language Processing",
    month = nov,
    year = "2025",
    address = "Suzhou, China",
    publisher = "Association for Computational Linguistics",
    url = "https://aclanthology.org/2025.emnlp-main.1332/",
    doi = "10.18653/v1/2025.emnlp-main.1332",
    pages = "26209--26222",
    ISBN = "979-8-89176-332-6"
}

@article{Gao2024RaBitQ,
author = {Gao, Jianyang and Long, Cheng},
title = {RaBitQ: Quantizing High-Dimensional Vectors with a Theoretical Error Bound for Approximate Nearest Neighbor Search},
year = {2024},
issue_date = {June 2024},
publisher = {Association for Computing Machinery},
address = {New York, NY, USA},
volume = {2},
number = {3},
url = {https://doi.org/10.1145/3654970},
doi = {10.1145/3654970},
journal = {Proc. ACM Manag. Data},
month = may,
articleno = {167},
numpages = {27}
}

@article{Ge2013OPQ,
  title={Optimized product quantization},
  author={Ge, Tiezheng and He, Kaiming and Ke, Qifa and Sun, Jian},
  journal={IEEE transactions on pattern analysis and machine intelligence},
  volume={36},
  number={4},
  pages={744--755},
  year={2013},
  publisher={IEEE}
}

@misc{harrier2026,
  title  = {harrier-oss-v1},
  author = {{Microsoft}},
  year   = {2026},
  note   = {Model card},
  url    = {https://huggingface.co/microsoft/harrier-oss-v1-27b},
}

@inproceedings{takeshita-etal-2025-randomly,
    title = "Randomly Removing 50{\%} of Dimensions in Text Embeddings has Minimal Impact on Retrieval and Classification Tasks",
    author = "Takeshita, Sotaro  and
      Takeshita, Yurina  and
      Ruffinelli, Daniel  and
      Ponzetto, Simone Paolo",
    booktitle = "Proceedings of the 2025 Conference on Empirical Methods in Natural Language Processing",
    month = nov,
    year = "2025",
    address = "Suzhou, China",
    publisher = "Association for Computational Linguistics",
    url = "https://aclanthology.org/2025.emnlp-main.1410/",
    doi = "10.18653/v1/2025.emnlp-main.1410",
    pages = "27705--27726",
    ISBN = "979-8-89176-332-6"
}

@inproceedings{Diera2025SoftZCA,
  author       = {Andor Diera and
                  Lukas Galke and
                  Ansgar Scherp},
  title        = {Isotropy Matters: Soft-ZCA Whitening of Embeddings for Semantic Code
                  Search},
  booktitle    = {33rd European Symposium on Artificial Neural Networks, Computational
                  Intelligence and Machine Learning, {ESANN} 2025, Bruges, Belgium,
                  April 23-25, 2025},
  year         = {2025},
  url          = {https://doi.org/10.14428/esann/2025.ES2025-58},
  doi          = {10.14428/ESANN/2025.ES2025-58},
  bibsource    = {dblp computer science bibliography, https://dblp.org}
}

@inproceedings{Faggioli2024DIME,
  author       = {Guglielmo Faggioli and
                  Nicola Ferro and
                  Raffaele Perego and
                  Nicola Tonellotto},
  title        = {Dimension Importance Estimation for Dense Information Retrieval},
  booktitle    = {Proceedings of the 47th International {ACM} {SIGIR} Conference on
                  Research and Development in Information Retrieval, {SIGIR} 2024, Washington
                  DC, USA, July 14-18, 2024},
  pages        = {1318--1328},
  publisher    = {{ACM}},
  year         = {2024},
  url          = {https://doi.org/10.1145/3626772.3657691},
  doi          = {10.1145/3626772.3657691},
  bibsource    = {dblp computer science bibliography, https://dblp.org}
}

@article{Faggioli2026GettingOffDIME,
  author       = {Guglielmo Faggioli and
                  Nicola Ferro and
                  Raffaele Perego and
                  Nicola Tonellotto},
  title        = {Getting off the {DIME:} Dimension Pruning via Dimension Importance
                  Estimation for Dense Information Retrieval},
  journal      = {{ACM} Trans. Inf. Syst.},
  volume       = {44},
  number       = {1},
  pages        = {9:1--9:34},
  year         = {2026},
  url          = {https://doi.org/10.1145/3765619},
  doi          = {10.1145/3765619},
  bibsource    = {dblp computer science bibliography, https://dblp.org}
}

@inproceedings{Campagnano2025UnveilingDIME,
  author       = {Cesare Campagnano and
                  Antonio Mallia and
                  Fabrizio Silvestri},
  title        = {Unveiling {DIME:} Reproducibility, Generalizability, and Formal Analysis
                  of Dimension Importance Estimation for Dense Retrieval},
  booktitle    = {Proceedings of the 48th International {ACM} {SIGIR} Conference on
                  Research and Development in Information Retrieval, {SIGIR} 2025, Padua,
                  Italy, July 13-18, 2025},
  pages        = {3367--3376},
  publisher    = {{ACM}},
  year         = {2025},
  url          = {https://doi.org/10.1145/3726302.3730318},
  doi          = {10.1145/3726302.3730318},
  bibsource    = {dblp computer science bibliography, https://dblp.org}
}

@inproceedings{DErasmo2025Eclipse,
  author       = {Giulio D'Erasmo and
                  Giovanni Trappolini and
                  Fabrizio Silvestri and
                  Nicola Tonellotto},
  title        = {Eclipse: Contrastive Dimension Importance Estimation with Pseudo-Irrelevance
                  Feedback for Dense Retrieval},
  booktitle    = {Proceedings of the 2025 International {ACM} {SIGIR} Conference on
                  Innovative Concepts and Theories in Information Retrieval, {ICTIR}
                  2025, Padua, Italy, 18 July 2025},
  pages        = {147--154},
  publisher    = {{ACM}},
  year         = {2025},
  url          = {https://doi.org/10.1145/3731120.3744579},
  doi          = {10.1145/3731120.3744579},
  bibsource    = {dblp computer science bibliography, https://dblp.org}
}

@inproceedings{d-erasmo-etal-2026-statistical,
    title = "Statistical Foundations of {DIME}: Risk Estimation for Practical Index Selection",
    author = "D'Erasmo, Giulio  and
      Campagnano, Cesare  and
      Mallia, Antonio  and
      Brutti, Pierpaolo  and
      Tonellotto, Nicola  and
      Silvestri, Fabrizio",
    booktitle = "Proceedings of the 19th Conference of the {E}uropean Chapter of the {A}ssociation for {C}omputational {L}inguistics (Volume 2: Short Papers)",
    month = mar,
    year = "2026",
    address = "Rabat, Morocco",
    publisher = "Association for Computational Linguistics",
    url = "https://aclanthology.org/2026.eacl-short.51/",
    doi = "10.18653/v1/2026.eacl-short.51",
    pages = "722--730",
    ISBN = "979-8-89176-381-4"
}

@inproceedings{wu-etal-2026-learning,
    title = "Learning to Select: Query-Aware Adaptive Dimension Selection for Dense Retrieval",
    author = "Wu, Zhanyu  and
      Zhang, Richong  and
      Nie, Zhijie",
    booktitle = "Proceedings of the 64th Annual Meeting of the {A}ssociation for {C}omputational {L}inguistics (Volume 1: Long Papers)",
    month = jul,
    year = "2026",
    address = "San Diego, California, United States",
    publisher = "Association for Computational Linguistics",
    url = "https://aclanthology.org/2026.acl-long.849/",
    doi = "10.18653/v1/2026.acl-long.849",
    pages = "18666--18677",
    ISBN = "979-8-89176-390-6"
}

@article{Fan2008FAIR,
author = {Jianqing Fan and Yingying Fan},
title = {{High-dimensional classification using features annealed independence rules}},
volume = {36},
journal = {The Annals of Statistics},
number = {6},
publisher = {Institute of Mathematical Statistics},
pages = {2605 -- 2637},
year = {2008},
doi = {10.1214/07-AOS504},
URL = {https://doi.org/10.1214/07-AOS504}
}

@article{Bickel2004Fisher,
author = {Peter J. Bickel and Elizaveta Levina},
title = {{Some theory for Fisher's linear discriminant function, `naive Bayes', and some alternatives when there are many more variables than observations}},
volume = {10},
journal = {Bernoulli},
number = {6},
publisher = {Bernoulli Society for Mathematical Statistics and Probability},
pages = {989 -- 1010},
year = {2004},
doi = {10.3150/bj/1106314847},
URL = {https://doi.org/10.3150/bj/1106314847}
}

@InProceedings{pmlr-v267-wen25e,
  title = 	 {Beyond Matryoshka: Revisiting Sparse Coding for Adaptive Representation},
  author =       {Wen, Tiansheng and Wang, Yifei and Zeng, Zequn and Peng, Zhong and Su, Yudi and Liu, Xinyang and Chen, Bo and Liu, Hongwei and Jegelka, Stefanie and You, Chenyu},
  booktitle = 	 {Proceedings of the 42nd International Conference on Machine Learning},
  pages = 	 {66520--66538},
  year = 	 {2025},
  volume = 	 {267},
  series = 	 {Proceedings of Machine Learning Research},
  month = 	 {13--19 Jul},
  publisher =    {PMLR},
  url = 	 {https://proceedings.mlr.press/v267/wen25e.html}
}

\appendix

\providecommand{\sigup}{\ensuremath{^{\uparrow}}}
\providecommand{\signs}{\ensuremath{^{\text{\tiny ns}}}}

\begin{table*}[t]
\centering
\caption{Retrieval performance on MS MARCO, NQ, HotpotQA, and FiQA. MS MARCO reports MRR@10 on the dev set; the other datasets report nDCG@10 on the test set. Scores are multiplied by 100 and averaged over seeds 1999, 5, and 2026 at target dimensions $k\in\{768,512,256,128\}$. Gray cells show full-dimensional references. \sigup{} marks a score significantly better than the original full dimension and \signs{} one not significantly different from it (paired t-test, $p<0.05$), in both cases consistently for all three seeds. Unmarked scores are significantly worse or inconsistent across seeds.}
\label{tab:main_results_beir}
\begingroup
\footnotesize
\setlength{\tabcolsep}{3.5pt}
\renewcommand{\arraystretch}{0.94}
\begin{tabular}{@{}llllllllllllllllllllllcccccccccccccccc@{}}
\toprule
\multirow{2}{*}{\textbf{Model}} & \multirow{2}{*}{\textbf{Method}}
& \multicolumn{4}{c}{\textbf{MS MARCO}}
& \multicolumn{4}{c}{\textbf{NQ}}
& \multicolumn{4}{c}{\textbf{HotpotQA}}
& \multicolumn{4}{c}{\textbf{FiQA}}
\\
\cmidrule(lr){3-6}
\cmidrule(lr){7-10}
\cmidrule(lr){11-14}
\cmidrule(lr){15-18}
& & \textbf{768} & \textbf{512} & \textbf{256} & \textbf{128} & \textbf{768} & \textbf{512} & \textbf{256} & \textbf{128} & \textbf{768} & \textbf{512} & \textbf{256} & \textbf{128} & \textbf{768} & \textbf{512} & \textbf{256} & \textbf{128} \\
\midrule

\multirow{14}{*}{Qwen3-8B}
& Full Dimension
& \multicolumn{4}{c}{\cellcolor{gray!25}{36.82}} & \multicolumn{4}{c}{\cellcolor{gray!25}{65.06}} & \multicolumn{4}{c}{\cellcolor{gray!25}{76.84}} & \multicolumn{4}{c}{\cellcolor{gray!25}{64.67}} \\
& PCA ($k{=}D$)
& \multicolumn{4}{c}{\cellcolor{gray!25}36.64\rlap{\signs}} & \multicolumn{4}{c}{\cellcolor{gray!25}65.33\rlap{\sigup}} & \multicolumn{4}{c}{\cellcolor{gray!25}76.54} & \multicolumn{4}{c}{\cellcolor{gray!25}64.31\rlap{\signs}} \\
& PCA ($k{=}D$) + RRC
& \multicolumn{4}{c}{\cellcolor{gray!25}36.66\rlap{\signs}} & \multicolumn{4}{c}{\cellcolor{gray!25}65.13\rlap{\signs}} & \multicolumn{4}{c}{\cellcolor{gray!25}76.61} & \multicolumn{4}{c}{\cellcolor{gray!25}64.26\rlap{\signs}} \\
& Prefix Truncation
& 36.43 & 35.67 & 34.52 & 32.40 & 63.72 & 63.18 & 61.07 & 57.29 & 75.27 & 74.10 & 70.67 & 62.83 & 63.95\signs & 63.46 & 61.59 & 57.36 \\
& Random Truncation
& 35.89 & 35.47 & 34.28 & 31.48 & 63.47 & 62.49 & 59.31 & 53.25 & 74.71 & 73.30 & 69.51 & 58.94 & 62.96 & 61.64 & 59.38 & 53.41 \\
& Random Projection
& 36.16 & 35.69 & 34.27 & 31.98 & 63.65 & 62.53 & 60.37 & 55.45 & 74.72 & 73.52 & 69.90 & 61.39 & 63.05 & 62.58 & 59.80 & 55.53 \\
& Matryoshka-Adaptor-U
& 36.52 & 35.86 & 34.52 & 32.26 & 63.80 & 63.35 & 61.21 & 57.31 & 75.32 & 74.11 & 70.85 & 63.42 & 63.86 & 63.44 & 61.38 & 57.10 \\
& Matryoshka-Adaptor-S
& 36.50 & 36.01 & 34.68 & 32.42 & 64.70 & 64.35 & 62.64 & 59.23 & 76.08 & 75.18 & 72.12 & 65.02 & 63.65 & 63.27 & 61.36 & 56.92 \\
\cdashline{2-18}
& PCA
& 36.02 & 35.49 & 34.14 & 31.29 & 64.63 & 63.84 & 62.05 & 57.20 & 74.90 & 73.40 & 68.75 & 59.62 & 63.77 & 63.68 & 62.10 & 59.34 \\
& PCA + RRC
& 36.09 & 35.49 & 34.20 & 31.37 & 64.74 & 63.98 & 62.19 & 57.26 & 74.84 & 73.36 & 68.62 & 59.48 & 63.78 & 63.39 & 61.80 & 58.88 \\
\cdashline{2-18}
& Whitening
& 34.94 & 35.19 & 34.12 & 32.33 & 61.87 & 62.33 & 61.79 & 58.30 & 73.70 & 73.62 & 70.17 & 61.29 & 58.37 & 59.77 & 60.53 & 58.99 \\
& Whitening + RRC
& 34.95 & 35.06 & 34.20 & 32.36 & 61.80 & 62.27 & 61.81 & 58.27 & 73.70 & 73.61 & 70.26 & 61.38 & 58.26 & 59.55 & 60.32 & 58.52 \\
\cdashline{2-18}
& SpecTemp
& 36.15 & 35.63 & 34.63 & 32.42 & 64.92\signs & 64.14 & 62.57 & 58.53 & 75.29 & 74.13 & 70.12 & 61.72 & 64.00 & 63.69 & 62.81 & 59.66 \\
& SpecTemp + RRC
& 36.18 & 35.67 & 34.65 & 32.46 & 64.84\signs & 64.12 & 62.67 & 58.67 & 75.40 & 74.25 & 70.38 & 61.84 & 63.80 & 63.65 & 62.61 & 59.55 \\
\midrule

\multirow{14}{*}{GTE-7B}
& Full Dimension
& \multicolumn{4}{c}{\cellcolor{gray!25}{39.14}} & \multicolumn{4}{c}{\cellcolor{gray!25}{66.93}} & \multicolumn{4}{c}{\cellcolor{gray!25}{73.06}} & \multicolumn{4}{c}{\cellcolor{gray!25}{62.01}} \\
& PCA ($k{=}D$)
& \multicolumn{4}{c}{\cellcolor{gray!25}39.05\rlap{\signs}} & \multicolumn{4}{c}{\cellcolor{gray!25}67.38\rlap{\sigup}} & \multicolumn{4}{c}{\cellcolor{gray!25}72.94} & \multicolumn{4}{c}{\cellcolor{gray!25}62.44\rlap{\signs}} \\
& PCA ($k{=}D$) + RRC
& \multicolumn{4}{c}{\cellcolor{gray!25}39.06\rlap{\signs}} & \multicolumn{4}{c}{\cellcolor{gray!25}67.29\rlap{\sigup}} & \multicolumn{4}{c}{\cellcolor{gray!25}72.94} & \multicolumn{4}{c}{\cellcolor{gray!25}63.56\rlap{\sigup}} \\

& Prefix Truncation
& 38.40 & 37.98 & 36.94 & 34.16 & 65.26 & 64.66 & 61.98 & 56.95 & 71.64 & 70.87 & 68.09 & 60.77 & 59.79 & 58.53 & 53.61 & 48.61 \\
& Random Truncation
& 38.36 & 37.88 & 36.50 & 34.08 & 65.30 & 64.37 & 61.38 & 56.90 & 71.66 & 70.62 & 67.70 & 61.17 & 60.34 & 59.38 & 56.65 & 50.28 \\
& Random Projection
& 38.31 & 37.84 & 36.69 & 34.38 & 65.32 & 64.40 & 62.49 & 57.33 & 71.48 & 70.58 & 68.06 & 62.05 & 60.57 & 59.70 & 56.72 & 52.09 \\
& Matryoshka-Adaptor-U
& 38.43 & 38.16 & 37.20 & 34.94 & 65.75 & 65.20 & 63.17 & 59.35 & 71.67 & 70.95 & 68.52 & 62.88 & 60.57 & 59.64 & 57.52 & 53.82 \\
& Matryoshka-Adaptor-S
& 38.38 & 38.07 & 37.01 & 34.88 & 66.50 & 66.15 & 64.60 & 61.21 & 73.29 & 72.76 & 70.76 & 65.74 & 59.43 & 58.55 & 56.34 & 52.60 \\
\cdashline{2-18}
& PCA
& 38.80 & 38.32 & 36.85 & 34.73 & 67.10\signs & 66.38 & 64.26 & 59.65 & 71.65 & 70.60 & 67.49 & 61.19 & 62.32\signs & 61.57\signs & 58.77 & 55.86 \\
& PCA + RRC
& 38.80 & 38.33 & 36.79 & 34.64 & 67.08\signs & 66.33 & 64.37 & 59.72 & 71.64 & 70.59 & 67.49 & 61.12 & 62.75 & 61.80\signs & 59.16 & 55.57 \\
\cdashline{2-18}
& Whitening
& 37.16 & 37.43 & 36.64 & 34.99 & 64.43 & 65.24 & 64.51 & 61.10 & 70.83 & 70.55 & 68.36 & 62.21 & 58.45 & 59.59 & 59.30 & 57.00 \\
& Whitening + RRC
& 37.17 & 37.41 & 36.65 & 34.96 & 64.33 & 65.21 & 64.48 & 61.12 & 70.86 & 70.54 & 68.46 & 62.25 & 58.28 & 59.33 & 59.16 & 56.99 \\
\cdashline{2-18}
& SpecTemp
& 38.84 & 38.44 & 37.02 & 35.11 & 67.25 & 66.81\signs & 65.11 & 61.27 & 71.98 & 71.06 & 68.38 & 62.30 & 62.54\signs & 62.31\signs & 60.41 & 57.41 \\
& SpecTemp + RRC
& 38.81 & 38.44 & 37.00 & 35.07 & 67.11\signs & 66.68\signs & 64.99 & 61.27 & 72.03 & 71.16 & 68.53 & 62.46 & 62.80 & 62.42\signs & 60.45 & 57.28 \\
\midrule

\multirow{14}{*}{DIVER-4B}
& Full Dimension
& \multicolumn{4}{c}{\cellcolor{gray!25}{25.75}} & \multicolumn{4}{c}{\cellcolor{gray!25}{54.33}} & \multicolumn{4}{c}{\cellcolor{gray!25}{65.87}} & \multicolumn{4}{c}{\cellcolor{gray!25}{47.78}} \\
& PCA ($k{=}D$)
& \multicolumn{4}{c}{\cellcolor{gray!25}25.94\rlap{\signs}} & \multicolumn{4}{c}{\cellcolor{gray!25}50.80} & \multicolumn{4}{c}{\cellcolor{gray!25}64.38} & \multicolumn{4}{c}{\cellcolor{gray!25}48.15\rlap{\signs}} \\
& PCA ($k{=}D$) + RRC
& \multicolumn{4}{c}{\cellcolor{gray!25}26.12} & \multicolumn{4}{c}{\cellcolor{gray!25}51.98} & \multicolumn{4}{c}{\cellcolor{gray!25}65.25} & \multicolumn{4}{c}{\cellcolor{gray!25}48.34\rlap{\signs}} \\
& Prefix Truncation
& 22.57 & 19.96 & 13.20 & 7.93 & 50.28 & 45.22 & 30.62 & 17.88 & 57.89 & 51.76 & 35.28 & 19.10 & 46.49 & 43.43 & 33.74 & 25.30 \\
& Random Truncation
& 24.46 & 23.53 & 20.66 & 15.84 & 50.95 & 48.99 & 43.15 & 32.08 & 61.33 & 58.42 & 49.62 & 33.34 & 45.50 & 44.00 & 39.36 & 31.48 \\
& Random Projection
& 23.98 & 23.32 & 20.80 & 16.09 & 50.38 & 48.69 & 43.17 & 32.82 & 60.65 & 58.08 & 49.51 & 33.66 & 45.06 & 43.76 & 40.36 & 32.87 \\
& Matryoshka-Adaptor-U
& 23.49 & 21.92 & 19.01 & 15.84 & 50.52 & 47.34 & 41.02 & 33.13 & 60.67 & 56.82 & 47.28 & 32.60 & 46.60 & 45.14 & 41.77 & 37.66 \\
& Matryoshka-Adaptor-S
& 25.56\signs & 24.79 & 23.08 & 20.30 & 54.69 & 52.55 & 48.56 & 42.49 & 65.98\signs & 63.43 & 57.34 & 46.56 & 41.56 & 38.60 & 34.34 & 31.21 \\
\cdashline{2-18}
& PCA
& 26.16\sigup & 25.50\signs & 23.88 & 20.68 & 49.79 & 48.51 & 44.45 & 36.70 & 61.42 & 58.58 & 51.08 & 39.65 & 47.64\signs & 47.21\signs & 44.68 & 41.29 \\
& PCA + RRC
& 26.40\sigup & 25.80\signs & 24.02 & 21.03 & 51.66 & 50.42 & 46.76 & 39.57 & 61.69 & 58.98 & 51.49 & 40.16 & 48.07\signs & 47.22\signs & 44.68 & 41.12 \\
\cdashline{2-18}
& Whitening
& 25.61\signs & 25.51\signs & 23.75 & 21.37 & 48.65 & 48.68 & 46.84 & 39.65 & 61.78 & 59.83 & 53.32 & 41.51 & 44.34 & 45.16 & 43.83 & 40.95 \\
& Whitening + RRC
& 25.79\signs & 25.72\signs & 24.03 & 21.58 & 49.19 & 49.33 & 47.79 & 41.33 & 62.37 & 60.61 & 54.33 & 42.52 & 44.04 & 45.07 & 43.63 & 40.51 \\
\cdashline{2-18}
& SpecTemp
& 26.28\sigup & 25.83\signs & 24.15 & 21.46 & 50.17 & 49.50 & 46.29 & 39.26 & 62.26 & 59.97 & 53.14 & 41.61 & 48.24\signs & 47.88\signs & 45.15 & 41.71 \\
& SpecTemp + RRC
& 26.54\sigup & 26.07\signs & 24.34 & 21.70 & 51.87 & 51.02 & 48.25 & 41.33 & 62.85 & 60.69 & 54.20 & 42.83 & 48.15\signs & 47.49\signs & 45.07 & 41.48 \\
\bottomrule
\end{tabular}
\endgroup
\end{table*}

\begin{table*}[t]
\centering
\caption{Retrieval performance (nDCG@10) on the BRIGHT datasets Biology, Earth Science, LeetCode, Economics, and TheoremQA-T. Scores are multiplied by 100 and averaged over seeds 1999, 5, and 2026 at target dimensions $k\in\{768,512,256,128\}$. Gray cells show full-dimensional references. \sigup{} marks a score significantly better than full dimension and \signs{} one not significantly different from it (paired t-test, $p<0.05$), in both cases consistently for all three seeds. Unmarked scores are significantly worse or inconsistent across seeds.}
\label{tab:main_results_bright}
\begingroup
\footnotesize
\setlength{\tabcolsep}{0.6pt}
\renewcommand{\arraystretch}{0.94}
\begin{tabular}{@{}lllllllllllllllllllllllllllcccccccccccccccccccc@{}}
\toprule
\multirow{2}{*}{{\textbf{Model}}} & \multirow{2}{*}{{\textbf{Method}}}
& \multicolumn{4}{c}{\textbf{Biology}}
& \multicolumn{4}{c}{\textbf{Earth Science}}
& \multicolumn{4}{c}{\textbf{LeetCode}}
& \multicolumn{4}{c}{\textbf{Economics}}
& \multicolumn{4}{c}{\textbf{TheoremQA-T}}
\\
\cmidrule(lr){3-6}
\cmidrule(lr){7-10}
\cmidrule(lr){11-14}
\cmidrule(lr){15-18}
\cmidrule(lr){19-22}
& & \textbf{768} & \textbf{512} & \textbf{256} & \textbf{128} & \textbf{768} & \textbf{512} & \textbf{256} & \textbf{128} & \textbf{768} & \textbf{512} & \textbf{256} & \textbf{128} & \textbf{768} & \textbf{512} & \textbf{256} & \textbf{128} & \textbf{768} & \textbf{512} & \textbf{256} & \textbf{128} \\
\midrule

\multirow{14}{*}{\rotatebox[origin=c]{90}{Qwen3-8B}}
& Full Dimension
& \multicolumn{4}{c}{\cellcolor{gray!25}{19.60}} & \multicolumn{4}{c}{\cellcolor{gray!25}{32.80}} & \multicolumn{4}{c}{\cellcolor{gray!25}{34.03}} & \multicolumn{4}{c}{\cellcolor{gray!25}{18.52}} & \multicolumn{4}{c}{\cellcolor{gray!25}{39.35}} \\
& PCA ($k{=}D$)
& \multicolumn{4}{c}{\cellcolor{gray!25}20.08\rlap{\signs}} & \multicolumn{4}{c}{\cellcolor{gray!25}31.18\rlap{\signs}} & \multicolumn{4}{c}{\cellcolor{gray!25}35.94\rlap{\signs}} & \multicolumn{4}{c}{\cellcolor{gray!25}21.93\rlap{\sigup}} & \multicolumn{4}{c}{\cellcolor{gray!25}40.60\rlap{\signs}} \\
& PCA ($k{=}D$) + RRC
& \multicolumn{4}{c}{\cellcolor{gray!25}25.17\rlap{\sigup}} & \multicolumn{4}{c}{\cellcolor{gray!25}28.57} & \multicolumn{4}{c}{\cellcolor{gray!25}34.42\rlap{\signs}} & \multicolumn{4}{c}{\cellcolor{gray!25}22.75\rlap{\sigup}} & \multicolumn{4}{c}{\cellcolor{gray!25}41.96} \\

& Prefix Truncation
& 18.20\signs & 17.53\signs & 16.98 & 14.96 & 31.95\signs & 30.19 & 28.58 & 28.58 & 32.95\signs & 34.07\signs & 32.58\signs & 27.32 & 19.96\signs & 19.55\signs & 19.32\signs & 15.36 & 35.72 & 34.87 & 31.90 & 27.32 \\
& Random Truncation
& 18.64 & 17.50 & 15.73 & 13.88 & 30.50 & 29.83 & 28.26 & 23.39 & 32.28 & 32.24 & 30.18 & 26.34 & 17.33\signs & 16.83\signs & 16.10 & 15.89 & 37.06 & 36.15 & 34.44 & 28.54 \\
& Random Projection
& 19.25\signs & 18.26\signs & 17.13 & 15.24 & 31.43 & 30.90 & 30.02 & 24.70 & 32.35 & 31.74 & 29.85 & 25.76 & 18.04\signs & 17.19\signs & 18.29\signs & 15.33 & 36.61 & 33.53 & 33.26 & 30.40 \\
& Matryoshka-Adaptor-U
& 18.51\signs & 17.74\signs & 17.17 & 15.29 & 31.06 & 30.45 & 29.38 & 27.35 & 34.71\signs & 35.18\signs & 33.98\signs & 28.15 & 20.38 & 20.68 & 19.55\signs & 17.11\signs & 35.95 & 34.83 & 31.94 & 27.82 \\
& Matryoshka-Adaptor-S
& 31.15\sigup & 31.78\sigup & 30.71\sigup & 29.59\sigup & 38.90 & 36.88 & 34.03\signs & 31.36\signs & 34.33\signs & 33.62\signs & 33.79\signs & 31.47\signs & 23.05\sigup & 22.18 & 20.42\signs & 16.51\signs & 42.39 & 41.22 & 37.63 & 32.72 \\
\cdashline{2-22}
& PCA
& 19.98\signs & 20.17\signs & 19.22\signs & 19.78\signs & 30.75 & 31.09\signs & 29.26 & 25.44 & 35.46\signs & 34.85\signs & 32.42\signs & 30.86\signs & 21.29\sigup & 21.05\sigup & 21.82\sigup & 20.98\signs & 40.82\sigup & 39.28\signs & 38.36\signs & 34.87 \\
& PCA + RRC
& 25.04\sigup & 25.56\sigup & 25.10\sigup & 24.01 & 29.04 & 28.32 & 25.54 & 22.38 & 34.75\signs & 33.97\signs & 31.69\signs & 29.04 & 22.73\sigup & 22.82\sigup & 21.64 & 19.17\signs & 41.67\signs & 40.85\signs & 37.92\signs & 33.14 \\
\cdashline{2-22}
& Whitening
& 16.91 & 18.69\signs & 17.37\signs & 16.73 & 27.78 & 29.52 & 29.82 & 25.78 & 33.72\signs & 33.25\signs & 31.14 & 29.47 & 18.02\signs & 19.37\signs & 21.28\signs & 17.22\signs & 41.73\signs & 40.60\signs & 39.20\signs & 35.34 \\
& Whitening + RRC
& 19.00\signs & 20.73\signs & 19.67\signs & 18.67\signs & 28.51 & 29.88 & 29.04 & 22.82 & 35.44\signs & 33.50\signs & 31.17 & 28.61 & 19.16\signs & 19.61\signs & 21.65 & 18.03\signs & 43.24 & 41.82\signs & 38.42\signs & 34.66 \\
\cdashline{2-22}
& SpecTemp
& 20.26\signs & 19.85\signs & 18.70\signs & 17.28 & 30.71 & 31.43\signs & 29.90 & 26.31 & 35.57\signs & 34.67\signs & 32.81\signs & 30.72\signs & 20.92\sigup & 20.63\signs & 21.01\signs & 18.99\signs & 41.54\sigup & 40.50\signs & 39.11\signs & 35.48 \\
& SpecTemp + RRC
& 24.62\sigup & 25.27\sigup & 24.51\sigup & 22.78 & 29.73 & 29.51 & 28.11 & 23.96 & 34.80\signs & 34.20\signs & 31.80\signs & 30.00 & 23.16\sigup & 23.41\sigup & 23.16\sigup & 20.29\signs & 41.98 & 41.10\signs & 39.04\signs & 35.51 \\
\midrule

\multirow{14}{*}{\rotatebox[origin=c]{90}{GTE-7B}}
& Full Dimension
& \multicolumn{4}{c}{\cellcolor{gray!25}{30.65}} & \multicolumn{4}{c}{\cellcolor{gray!25}{40.18}} & \multicolumn{4}{c}{\cellcolor{gray!25}{30.77}} & \multicolumn{4}{c}{\cellcolor{gray!25}{15.81}} & \multicolumn{4}{c}{\cellcolor{gray!25}{35.75}} \\
& PCA ($k{=}D$)
& \multicolumn{4}{c}{\cellcolor{gray!25}25.51} & \multicolumn{4}{c}{\cellcolor{gray!25}36.40} & \multicolumn{4}{c}{\cellcolor{gray!25}31.66\rlap{\signs}} & \multicolumn{4}{c}{\cellcolor{gray!25}13.54} & \multicolumn{4}{c}{\cellcolor{gray!25}37.03\rlap{\signs}} \\
& PCA ($k{=}D$) + RRC
& \multicolumn{4}{c}{\cellcolor{gray!25}32.52\rlap{\signs}} & \multicolumn{4}{c}{\cellcolor{gray!25}40.48\rlap{\signs}} & \multicolumn{4}{c}{\cellcolor{gray!25}32.49} & \multicolumn{4}{c}{\cellcolor{gray!25}20.67} & \multicolumn{4}{c}{\cellcolor{gray!25}47.12\rlap{\sigup}} \\
& Prefix Truncation
& 26.29 & 24.60 & 19.50 & 13.48 & 39.48\signs & 38.59\signs & 33.77 & 26.02 & 28.80\signs & 26.86 & 23.50 & 17.75 & 16.63\signs & 17.52\signs & 13.46\signs & 8.42 & 34.96\signs & 30.59 & 22.79 & 16.66 \\
& Random Truncation
& 28.63\signs & 26.79 & 22.82 & 16.13 & 37.12 & 35.72 & 31.64 & 26.68 & 29.27 & 28.04 & 24.77 & 19.70 & 13.98 & 13.89 & 11.60 & 8.54 & 29.91 & 29.22 & 22.29 & 17.16 \\
& Random Projection
& 26.99 & 25.94 & 23.73 & 16.81 & 38.34\signs & 36.24 & 34.44 & 28.39 & 28.90 & 27.11 & 23.97 & 20.56 & 15.68\signs & 14.84\signs & 11.60 & 9.08 & 33.60 & 31.86 & 24.04 & 17.87 \\
& Matryoshka-Adaptor-U
& 26.95 & 26.09 & 22.44 & 18.00 & 39.15 & 38.45\signs & 34.35 & 28.01 & 27.66 & 26.41 & 23.23 & 18.36 & 14.74\signs & 15.59\signs & 11.59 & 9.01 & 33.19 & 27.84 & 21.37 & 17.15 \\
& Matryoshka-Adaptor-S
& 30.76\signs & 30.09\signs & 27.51\signs & 26.05 & 29.48 & 28.45 & 26.18 & 22.89 & 17.37 & 17.18 & 16.23 & 14.94 & 4.03 & 4.12 & 2.86 & 2.47 & 44.80 & 45.14 & 43.73 & 41.52\signs \\
\cdashline{2-22}
& PCA
& 22.88 & 20.78 & 18.21 & 15.70 & 34.15 & 33.12 & 28.17 & 22.12 & 31.69\signs & 31.19\signs & 29.11\signs & 25.54 & 11.85 & 11.47 & 10.58 & 7.83 & 36.83\signs & 36.26\signs & 32.46\signs & 25.03 \\
& PCA + RRC
& 28.80\signs & 27.38 & 22.61 & 19.27 & 36.84 & 34.52 & 30.37 & 25.96 & 31.89\signs & 31.61\signs & 27.87 & 25.42 & 19.22 & 17.86 & 16.64\signs & 12.55 & 46.41\sigup & 45.34\sigup & 40.47 & 31.46\signs \\
\cdashline{2-22}
& Whitening
& 26.90\signs & 27.01 & 24.36 & 23.95 & 38.67\signs & 38.81\signs & 32.39 & 23.93 & 27.61 & 28.68\signs & 27.57 & 25.01 & 18.13\signs & 18.70\signs & 17.64\signs & 13.54\signs & 33.11\signs & 31.88\signs & 33.18\signs & 27.85 \\
& Whitening + RRC
& 28.50\signs & 28.03\signs & 24.98 & 24.42 & 40.88\signs & 39.14\signs & 33.77 & 27.99 & 27.74 & 29.02\signs & 26.73 & 24.71 & 19.56\signs & 19.92\signs & 18.54\signs & 15.00\signs & 39.52\signs & 39.79\signs & 41.08 & 36.25\signs \\
\cdashline{2-22}
& SpecTemp
& 25.20 & 23.75 & 22.50 & 22.15 & 34.78 & 33.79 & 30.06 & 23.54 & 31.08\signs & 30.96\signs & 28.68\signs & 25.69 & 12.29 & 12.47 & 13.52\signs & 13.29\signs & 37.34\signs & 37.67\signs & 34.94\signs & 28.76 \\
& SpecTemp + RRC
& 30.62\signs & 28.58\signs & 24.59 & 23.12 & 37.49 & 35.48 & 32.80 & 28.19 & 31.71\signs & 31.47\signs & 28.35\signs & 25.94 & 20.03 & 19.38 & 18.27\signs & 15.08\signs & 47.01\sigup & 46.17\sigup & 43.95\sigup & 38.09\signs \\
\midrule

\multirow{14}{*}{\rotatebox[origin=c]{90}{DIVER-4B}}
& Full Dimension
& \multicolumn{4}{c}{\cellcolor{gray!25}{42.49}} & \multicolumn{4}{c}{\cellcolor{gray!25}{44.63}} & \multicolumn{4}{c}{\cellcolor{gray!25}{37.40}} & \multicolumn{4}{c}{\cellcolor{gray!25}{21.78}} & \multicolumn{4}{c}{\cellcolor{gray!25}{36.09}} \\
& PCA ($k{=}D$)
& \multicolumn{4}{c}{\cellcolor{gray!25}34.04} & \multicolumn{4}{c}{\cellcolor{gray!25}39.67} & \multicolumn{4}{c}{\cellcolor{gray!25}39.04\rlap{\signs}} & \multicolumn{4}{c}{\cellcolor{gray!25}19.35} & \multicolumn{4}{c}{\cellcolor{gray!25}37.32\rlap{\signs}} \\
& PCA ($k{=}D$) + RRC
& \multicolumn{4}{c}{\cellcolor{gray!25}33.81} & \multicolumn{4}{c}{\cellcolor{gray!25}43.03\rlap{\signs}} & \multicolumn{4}{c}{\cellcolor{gray!25}37.91\rlap{\signs}} & \multicolumn{4}{c}{\cellcolor{gray!25}18.23} & \multicolumn{4}{c}{\cellcolor{gray!25}37.85\rlap{\signs}} \\

& Prefix Truncation
& 43.54\signs & 41.74\signs & 34.26 & 25.66 & 42.54\signs & 41.17 & 34.01 & 21.73 & 36.38\signs & 33.84 & 28.06 & 20.35 & 21.29\signs & 18.80 & 15.00 & 11.29 & 35.36\signs & 34.42\signs & 29.28 & 22.23 \\
& Random Truncation
& 38.69 & 37.22 & 33.19 & 23.84 & 42.80 & 41.17 & 37.01 & 28.81 & 37.02\signs & 36.29 & 32.88 & 28.20 & 21.36\signs & 20.46\signs & 18.72 & 13.83 & 35.46\signs & 34.40\signs & 31.72 & 24.37 \\
& Random Projection
& 40.95 & 36.82 & 31.77 & 27.26 & 42.58 & 42.48 & 38.55 & 31.58 & 36.12 & 34.69 & 32.61 & 27.45 & 21.65\signs & 20.69\signs & 17.71 & 16.16 & 34.90\signs & 33.70 & 30.26 & 23.26 \\
& Matryoshka-Adaptor-U
& 40.69\signs & 38.52 & 35.27 & 30.98 & 42.82 & 42.51 & 38.96 & 35.61 & 36.74\signs & 35.23 & 33.56 & 28.88 & 21.23\signs & 20.37 & 18.86 & 15.18 & 36.47\signs & 36.28\signs & 32.33\signs & 28.06 \\
& Matryoshka-Adaptor-S
& 21.18 & 19.99 & 18.90 & 17.56 & 27.49 & 26.00 & 23.63 & 21.12 & 35.85\signs & 35.47\signs & 34.51\signs & 33.10\signs & 12.76 & 12.19 & 10.72 & 9.30 & 45.31 & 46.79 & 44.44 & 43.77 \\
\cdashline{2-22}
& PCA
& 32.18 & 31.78 & 30.10 & 27.45 & 36.82 & 34.72 & 33.16 & 29.21 & 38.17\signs & 37.49\signs & 35.77\signs & 32.64 & 18.80 & 19.02 & 17.24 & 15.33 & 36.91\signs & 35.52\signs & 34.76\signs & 32.39\signs \\
& PCA + RRC
& 33.06 & 32.53 & 31.87 & 28.63 & 40.09 & 39.00 & 36.41 & 32.27 & 37.41\signs & 36.94\signs & 33.46 & 27.94 & 18.06 & 18.08 & 16.62 & 14.64 & 37.97\signs & 37.57\signs & 36.38\signs & 32.30\signs \\
\cdashline{2-22}
& Whitening
& 30.12 & 29.73 & 29.00 & 27.22 & 37.17 & 35.81 & 32.71 & 30.17 & 37.25\signs & 37.77\signs & 34.76 & 32.66 & 18.34 & 19.49\signs & 16.50 & 17.09 & 34.63\signs & 36.50\signs & 37.23\signs & 34.01\signs \\
& Whitening + RRC
& 29.31 & 28.41 & 29.62 & 28.38 & 36.05 & 34.99 & 31.97 & 29.26 & 37.51\signs & 38.72\signs & 33.94 & 28.27 & 18.78 & 17.88 & 15.15 & 14.05 & 33.39\signs & 35.10\signs & 34.41\signs & 32.46\signs \\
\cdashline{2-22}
& SpecTemp
& 32.18 & 31.56 & 30.79 & 27.58 & 36.79 & 35.42 & 33.62 & 30.33 & 38.15\signs & 37.40\signs & 35.02 & 32.97 & 19.16 & 19.23 & 17.86 & 16.46 & 36.81\signs & 36.95\signs & 36.63\signs & 34.67\signs \\
& SpecTemp + RRC
& 33.09 & 32.52 & 32.01 & 29.06 & 39.80 & 39.14 & 35.51 & 30.75 & 37.81\signs & 37.33\signs & 34.59 & 29.13 & 18.01 & 18.16 & 16.50 & 14.92 & 38.11\signs & 37.65\signs & 36.73\signs & 33.10\signs \\
\bottomrule
\end{tabular}
\endgroup
\end{table*}

\section{Methodological Details and Proofs}
\label{app:methodology}

We provide the exponent-selection rule used by SpecTemp and the proofs of the two propositions in the methodology.

\subsection{SpecTemp Exponent Selection}
\label{app:spectemp}

SpecTemp~\cite{LiEK26_SpecTemp} selects the power-law exponent in Eq.~\eqref{eq:priors} from the corpus spectrum and the target dimension. Let $\mathcal{T}$ index the smallest decile of the full spectrum. SpecTemp estimates a noise floor and defines a spectral signal-to-noise ratio as
\begin{equation}
    \sigma^2_{\mathrm{noise}}
    =\frac{1}{|\mathcal{T}|}\sum_{j\in\mathcal{T}}\lambda_j,
    \quad
    \operatorname{SNR}(i)
    =\left(
        \frac{\lambda_i}{\sigma^2_{\mathrm{noise}}}-1
      \right)_{+}.
    \label{eq:noise-floor}
\end{equation}
It then sets
\begin{equation}
    \gamma_{\mathrm{ST}}(k)
    =\min\left\{
        1,
        \frac{\operatorname{SNR}(k)}
             {\operatorname{SNR}(k_{\mathrm{knee}})}
      \right\},
    \label{eq:gamma-st}
\end{equation}
where $k_{\mathrm{knee}}$ is the knee of the SNR curve detected by Kneedle~\cite{SatopaaAIR11_Kneedle}. Since the eigenvalues are ordered in descending order, this rule applies full whitening up to the knee and progressively weaker inverse-variance scaling at larger target dimensions.

\subsection{Proof of Proposition~\ref{prop:margin-optimum}}
\label{app:margin-optimum-proof}

\begin{proof}
For any nonnegative, nonzero weight vector $\mathbf{w}$,
replacing $m_i$ by $(m_i)_{+}$ cannot decrease the numerator
of $A(\mathbf{w})$. Therefore,
\begin{equation}
    A(\mathbf{w})
    \leq
    \frac{\sum_{i=1}^{k}w_i(m_i)_{+}}
         {\sqrt{\sum_{i=1}^{k}w_i^2v_i}}
    \leq
    \sqrt{\sum_{i=1}^{k}\frac{(m_i)_{+}^{2}}{v_i}}.
    \label{eq:margin-bound}
\end{equation}
The second inequality follows from the Cauchy--Schwarz
inequality applied to the vectors with entries
$\sqrt{v_i}\,w_i$ and $(m_i)_{+}/\sqrt{v_i}$.

Both vectors are nonzero under the assumptions of the
proposition. Equality in the second inequality holds
if and only if they are positive multiples of each other,
giving
\[
    w_i=c\,\frac{(m_i)_{+}}{v_i},
    \qquad c>0.
\]
These weights assign zero to every direction with
$m_i\leq0$, so they also attain equality in the first
inequality. They are therefore precisely the maximisers,
and the maximum is
$\sqrt{\sum_{i=1}^{k}(m_i)_{+}^{2}/v_i}$.
\end{proof}

\subsection{Proof of Proposition~\ref{prop:residual-invariance}}
\label{app:residual-invariance-proof}

\begin{proof}
Let tildes denote quantities computed from the rescaled coordinates
$\widetilde{c}_i(t)=\lambda_i^\alpha c_i(t)$ of both queries and documents.
Then $\widetilde{z}_{qj,i}=\lambda_i^{2\alpha}z_{qj,i}$ for every training pair, so
\[
    \widetilde{\widehat{m}}_i=\lambda_i^{2\alpha}\widehat{m}_i,
    \qquad
    \widetilde{\widehat{v}}_i=\lambda_i^{4\alpha}\widehat{v}_i,
    \qquad
    \widetilde{\widehat{e}}_i=\lambda_i^{-2\alpha}\widehat{e}_i .
\]
Since $\lambda_i>0$, the positive support $\mathcal{I}_{+}$ is unchanged.
On $\mathcal{I}_{+}$, Eq.~\eqref{eq:residual} gives
\[
    \log\widetilde{\widehat{e}}_i
    =a-(\gamma_{\mathrm e}+2\alpha)\log\lambda_i+u_i .
\]
By Eq.~\eqref{eq:residual-orthogonality}, the residual vector $\mathbf u=(u_i)_{i\in\mathcal I_+}$ is orthogonal to both regressors in Eq.~\eqref{eq:powerlaw-fit}, the constant and $\log\lambda_i$.
Thus, the least-squares fit to the rescaled log weights retains intercept $a$ and shifts the exponent to $\gamma_{\mathrm e}+2\alpha$, leaving the residual $u_i$ unchanged.
\end{proof}

\section{Implementation Details}
\label{app:implementation-details}

\paragraph{Embedding generation and spectral estimation.}
We generate and cache query and document embeddings using the checkpoints listed in Section~\ref{sec:models}, keeping all encoder parameters fixed. The same cached embeddings are reused across compression methods.
We $\ell_2$-normalise the original embeddings before fitting the compression transformations. For each dataset and seed, we estimate the corpus mean and centred covariance from up to one million documents sampled uniformly without replacement, using the full corpus when it is smaller. The resulting eigendecomposition provides the principal directions and eigenvalues used by the spectral methods. These corpus-derived quantities are estimated without relevance labels and shared across query folds. Each RRC variant uses the same spectral basis as its corresponding uncalibrated compressor. After applying a compression transformation, we $\ell_2$-normalise the output vectors for cosine retrieval.

\paragraph{Residual Calibration}

For each training query, we pair every judged positive with one sampled negative. For each query, we cache the top-200 candidates retrieved using the
full-dimensional embeddings; these label-free retrieval pools are shared across folds. On BRIGHT, excluded documents are removed before truncating the pools to 200 and are also excluded from random-negative sampling and evaluation rankings. 
A fixed majority of negatives are drawn uniformly from the query's pool and the rest from the corpus, excluding judged positives from both sources.
We fall back to corpus sampling when no eligible pool document remains. Triples are fixed within each training fold and sampled independently across folds.
We compute margin contributions from centred, unweighted spectral projections, before output normalisation. We average margins and their squares within each query and then across queries. For each direction, the variance estimate is the averaged squared margin minus the square of the averaged margin. We obtain the margin-based weights, fit their power-law trend on the positive support, and add the residual to the prior's log weights as described in Section~\ref{sec:residual-calibration}. The resulting transformation is fixed before evaluation and applied identically to the document corpus and held-out queries.\looseness=-1

\paragraph{Matryoshka-Adaptor}

We reimplement Matryoshka-Adaptor~\cite{yoon-etal-2024-matryoshka} based on the paper, as no official implementation is publicly available.
We use a residual MLP $D\!\rightarrow\!1024\!\rightarrow\!D$ with ReLU and a skip connection, initialising the output layer to zero to start from the identity map.
The same adaptor transforms queries and documents, followed by prefix truncation and $\ell_2$ normalisation for cosine retrieval. Both variants jointly optimise $k\in\{128,256,512,768\}$, averaging similarity-based losses across prefix lengths and assigning unit coefficients to all loss terms.

The unsupervised variant combines pairwise-similarity preservation between distinct batch documents, neighbourhood-similarity preservation, and full-dimensional reconstruction using absolute deviations from the original embeddings.
Each document's top-10 neighbours are retrieved from a fixed random sample of at most 50,000 corpus documents. Training uses 5,000 steps at the reported learning rate of $10^{-3}$.

The supervised variant starts from the trained unsupervised adaptor and is fine-tuned on the same training folds as RRC; held-out queries enter none of its losses.
It applies the unsupervised objectives to corpus documents and training queries, using each query's top-10 full-dimensional retrieval results as its neighbours.
At each prefix length, it adds a pairwise logistic ranking loss weighted by the positive document's relevance grade.

Training triples are resampled at each step by drawing a training query, one of its judged positives, and a negative.
Negative pools are formed by removing judged positives and BRIGHT exclusions from the top-200 full-dimensional retrieval results, without replenishment.
The mixture of pool and corpus negatives, exclusions for random negatives, and fallback to corpus sampling match RRC.

Both stages use Adam with a batch size of 128. In preliminary held-out evaluations on BRIGHT, supervised fine-tuning with the originally reported learning rate of $10^{-3}$ underperformed the unsupervised adaptor.
Based on these results, we reduced the supervised learning rate to $10^{-5}$ for all configurations, with no further hyperparameter search.
We train for 10,000 steps, compared with the original maximum of 5,000. The original paper reports early stopping; because each BRIGHT training fold provides only 60--114 queries, we reserve no development set and evaluate the final checkpoint. RRC uses the same training folds and requires neither a learning rate nor a stopping criterion.

\paragraph{Computing environment.}
Spectral estimation and residual calibration are implemented in NumPy. We use Faiss for exhaustive retrieval and for the quantization experiments described in Section~\ref{sec:exp-quantization}.
Experiments are conducted on a cluster equipped with four NVIDIA H100 GPUs.

\end{document}